\documentclass[10pt,journal,compsoc, peerreview]{IEEEtran}
\ifCLASSOPTIONcompsoc
  \usepackage[nocompress]{cite}
\else
  \usepackage{cite}
\fi
\ifCLASSINFOpdf
\else
\fi
\usepackage{graphicx}
\usepackage{balance}  

\usepackage{times}
\usepackage{mathtools}

\usepackage{subfigure}
\usepackage{amsmath,amssymb}
\usepackage{amsthm}
\usepackage{thmtools}
\usepackage[ruled,linesnumbered]{algorithm2e}
\usepackage{xcolor}
\usepackage{paralist}
\usepackage{float}
\usepackage{url}
\usepackage{hyperref}
\newtheorem{definition}{Definition}
\newtheorem{lemma}{Lemma}

\newtheorem{theorem}{Theorem}
\newtheorem{example}{Example}

\begin{document}
%
\title{Efficient Information Flow Maximization in Probabilistic Graphs}
%
%
%
%

\author{Christian Frey,
        Andreas~Z\"{u}fle,
        Tobias Emrich
        and Matthias Renz\vspace{-0.3cm}
\IEEEcompsocitemizethanks{\IEEEcompsocthanksitem C.Frey and T.Emrich are with
the Department of Database Systems and Data Mining,
Ludwig-Maximilians-Universit\"{a}t, Munich, Germany.\protect\\
E-mail: \{frey, emrich\}@dbs.ifi.lmu.de
\IEEEcompsocthanksitem A.Z\"{u}fle is with the Department of Geography and Geoinformation
Science, George Mason University, VA, USA
E-mail: azufle@gmu.edu
\IEEEcompsocthanksitem M.Renz is with the Department of Computational and Data
Sciences, George Mason University, VA, USA
E-mail: mrenz@gmu.edu
}}
%
%

\markboth{Journal of IEEE Transactions on Knowledge and Data Engineering (TKDE)}%
{Shell \MakeLowercase{\textit{et al.}}: Bare Demo of IEEEtran.cls for Computer Society Journals}
%



\IEEEtitleabstractindextext{%
\begin{abstract}
Reliable propagation of information through large networks, e.g., communication
networks, social networks or sensor networks is very important in many
applications concerning marketing, social networks, and wireless sensor
networks.
However, social ties of friendship may be obsolete, and communication links may fail, inducing the notion of uncertainty in such networks.
In this paper, we address the problem of optimizing information propagation in
uncertain networks given a constrained budget of edges. We show that this
problem requires to solve two NP-hard subproblems: the computation of expected
information flow, and the optimal choice of edges. To compute the expected
information flow to a source vertex, we propose the \emph{F-tree} as a
specialized data structure, that identifies independent components of the
graph for which the information flow can either be computed analytically and efficiently, or for which traditional Monte-Carlo sampling can be applied independently of the remaining network.
For the problem of finding the optimal edges, we propose a series of heuristics
that exploit properties of this data structure. Our evaluation shows
that these heuristics lead to high quality solutions, thus yielding high
information flow, while maintaining low running time. \vspace{-0.25cm} 

\end{abstract} 

\begin{IEEEkeywords}
uncertain graphs, network analysis, social network, optimization, information
flow\vspace{-0.1cm}
\end{IEEEkeywords}}

\maketitle
\def\vertexSet{\ensuremath{V_G}}
\def\edgeSet{\ensuremath{E_G}}
\def\vertexSink{\ensuremath{v_s}}
\def\vertex{\ensuremath{v_i}}
\def\edge{\ensuremath{e_{v_i, v_j}}}
\def\predecessor{\ensuremath{pred(\vertex)}}
\def\successor{\ensuremath{succ(\vertex)}}
\def\probabilityFunc{\ensuremath{p(\edge)}}
\def\weightFunc{\ensuremath{w(v_i)}}
\def\neighborSet{\ensuremath{N(\vertex)}}
\def\graphDef{\ensuremath{G(V,E,p,w)}}
\def\path{\ensuremath{path_{G}(v_i, v_j)}}
\def\informationFlowFunc{\ensuremath{if(v_i)}}

\newcommand{\argmax}{\operatornamewithlimits{argmax}}

\IEEEdisplaynontitleabstractindextext

\ifCLASSOPTIONpeerreview
\begin{center} \bfseries To refer to or cite this work, please use the citation of the published version: \\
\vspace{.5cm}
\hyperlink{http://ieeexplore.ieee.org/document/8166795/}{http://ieeexplore.ieee.org/document/8166795/}\\ \vspace{.5cm}
C. Frey, A. Z\"ufle, T. Emrich and M. Renz, "Efficient Information Flow Maximization in Probabilistic Graphs," in IEEE Transactions on Knowledge \& Data Engineering, vol. 30, no. 5, pp. 880-894, 2018.
doi:10.1109/TKDE.2017.2780123
\end{center}
\fi
%
\IEEEpeerreviewmaketitle

\IEEEraisesectionheading{\section{Introduction}\label{sec:intro}}
\vspace{-0.25cm}
\IEEEPARstart{N}{owadays}, social and communication networks have become ubiquitous in our daily life to receive and share information. Whenever we are navigating the World Wide Web, updating our social network profiles, or sending a text message on our cell-phone, we participate in an information network as a node. In such settings, network nodes exchange some sort of information: In social networks, users share their opinions and ideas, aiming to convince others. In wireless sensor networks, nodes collect data and aim to ensure that this data is propagated through the network: Either to a destination, such as a server node, or simply to as many other nodes as possible. Abstractly speaking, in all of these networks, nodes aim at propagating their information, or their belief, throughout the network. The event of a successful propagation of information between nodes is subject to inherent uncertainty. In a wireless sensor, telecommunication or electrical network, a link can be unreliable and may fail with certain probability \cite{routing2007,rubino1998network}. In a social network, trust and influence issues may impact the likelihood of social interactions or the likelihood of convincing another of an individual's idea \cite{guha2004propagation,kempe2003maximizing,adar2007managing}. For example, consider professional social networks like \emph{LinkedIn}. Such networks allow users to endorse each others' skills and abilities. Here, the probability of an edge may reflect the likelihood that one user is willing to endorse another user. The probabilistic graph model is commonly used to address such scenarios in a unified way (e.g. \cite{li2012mining,papapetrou2011efficient,kneighbors2010,zou2010finding,yuan2013efficient,kollios2013clustering, zou2010mining}). In this model, each edge is associated with an existential probability to quantify the likelihood that this edge exists in the graph.
Traditionally, to maximize the likelihood of a successful communication between two nodes, information is propagated by flooding it through the network. Thus, every node that receives a bit of information will proceed to share this information with all its neighbors. Clearly, such a flooding approach is not applicable for large communication and social networks, as the communication between two network nodes incurs a cost: Sensor network nodes have limited computing capability, memory resources and power supply, but require battery power to send and receive messages, and are also limited by their bandwidth; individuals of a social network require time and sometimes even additional monetary resources to convince others of their ideas.
For instance, a professional networking service may provide, for a fee, a service to directly ask a limited number of users to endorse another user $Q$. The challenge is to maximize the expected number of endorsements that $Q$ will receive, while limiting the budget of users asked by the service provider. The first candidates to ask are $Q$'s direct connections. In addition, if a user $u$ has already endorsed $Q$, then $u$'s connections can be asked if they trust $u$'s judgment and want to make the same endorsement.

\begin{figure}[t!]\vspace{-0.8cm}
    \subfigure[original graph]{
        \includegraphics[width =
        0.45\columnwidth]{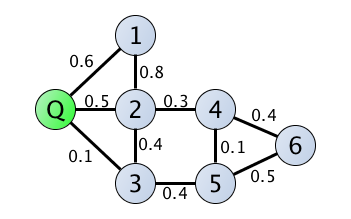}
		\label{fig:runningExample_a}
   } \subfigure[Dijkstra MST]{
       \includegraphics[width =
       0.45\columnwidth]{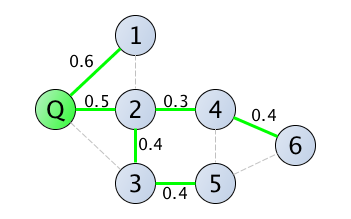}
       \label{fig:runningExample_b}
   }
\vspace{-0.2cm}\\
	\subfigure[Optimal five-edge flow]{
       \includegraphics[width =
       0.45\columnwidth]{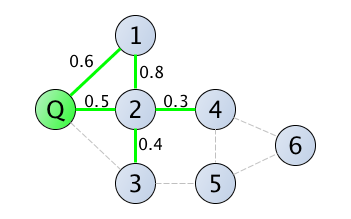}
		\label{fig:runningExample_c}
   } \subfigure[possible world $g_1$]{
       \includegraphics[width =
       0.45\columnwidth]{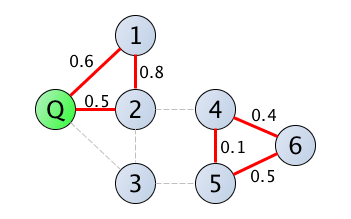}
       \label{fig:runningExample_d}
   } 
	\vspace{-0.2cm}
    \caption{Running example.}\vspace{-0.5cm}
	\label{fig:runningExample}
\end{figure}
In this work, we address the following problem: Given a probabilistic network graph $\mathcal{G}$ with edges that can be activated, i.e., enabled to transfer information, or stay inactive. The problem is to send/receive information from a single node $Q$ in $\mathcal{G}$ to/from as many nodes in $\mathcal{G}$ as possible assuming a limited budget of edges that can be activated. 
To solve this problem, the main focus is on the selection of edges to be activated.
\begin{example}\label{ex:runex}To illustrate our problem setting, consider the network depicted in Figure \ref{fig:runningExample_a}. The task is to maximize the information flow to node $Q$ from other nodes given a limited budget of edges. This example assumes equal weights of all nodes. Each edge of the network is labeled with the probability of a successful communication. A naive solution is to activate all edges. Assuming each node to have one unit of information, the expected information flow of this solution can be shown to be $\simeq 2.51$. While maximizing the information flow, this solution incurs the maximum possible communication cost. A traditional trade-off between these single-objective solutions is using a probability maximizing \emph{Dijkstra's} MST, as depicted in Figure \ref{fig:runningExample_b}. The expected information flow in this setting can be shown to aggregate to $1.59$ units, while requiring six edges to be activated.
Yet, it can be shown that the solution depicted in Figure \ref{fig:runningExample_c} dominates this solution: Only fives edges are used, thus further reducing the communication cost, while achieving a higher expected information flow of $\simeq 2.02$ units of information to $Q$.
\end{example}

The aim of this work is to efficiently find a near-optimal sub-network, which maximizes the expected flow of information at a constrained budget of edges. In Example \ref{ex:runex}, we computed the information flow for an example graph. But in fact, this computation has been shown to be exponentially hard in the number of edges of the graph, and thus impractical to be solved analytically. Furthermore, the optimal selection of edges to maximize the information flow is shown to be \emph{NP}-hard. These two subproblems define the main computational challenges addressed in this work.

To tackle these challenges, the remainder of this work is organized as follows. After a survey of related work in Section \ref{sec:related}, we recapitulate common definitions for stochastic networks and formally define our problem setting in Section \ref{sec:fundamentals}. After a more detailed technical overview in Section \ref{sec:roadmap}, the theoretical heart of this work is presented in Section \ref{sec:reachability}. We show how to identify independent subgraphs, for which the information flow can be computed independently. This allows to divide the main problem into much smaller subproblems. To conquer these subproblems, we identify cases for which the expected information flow can be computed analytically, and we propose to employ Monte-Carlo sampling to approximate the information flow of the remaining cases. Section \ref{sec:ct} is the algorithmic core of our work, showing how aforementioned independent components can be organized hierarchically in a \emph{F-tree} which is inspired by the \emph{block-cut tree} \cite{tarjan1972depth,hopcroft1973algorithm,westbrook1992maintaining}. This structure allows us to aggregate results of individual components efficiently, and we show how previous Monte-Carlo sampling results can be re-used as more edges are selected and activated. Our experimental evaluation in Section \ref{sec:evaluation} shows that our algorithms significantly outperform traditional solutions, in terms of combined communication cost and information flow, on synthetic and real stochastic networks. In summary, the main contributions of this work are:
\begin{compactitem}
\item Theoretical complexity study of the flow maximization problem in probabilistic graphs.
\item Efficient estimation of the expected information flow based on network graph decomposition and Monte-Carlo sampling.
\item Our \emph{F-tree} structure enabling efficient
organization of independent graph components and (local) intermediate results for efficient expected flow computation.
\item An algorithm for iterative selection of edges to be activated to maximize the expected information flow.
\item Thorough experimental evaluation of proposed algorithms.
\end{compactitem}

%
%
%
%

\section{Related Work}
\label{sec:related}

Reliability and Influence computation in probabilistic graphs (a.k.a. uncertain graphs) has recently attracted much attention in the data mining and database research communities.
We summarize state-of-the-art publications and relate our work in this context.

{\bf Subgraph Reliability.} A related and fundamental problem in uncertain graph mining is the so-called subgraph reliability problem, which asks to estimate the probability that two given (sets of) nodes are reachable. This problem, well studied in the context of communication networks, has seen a recent revival in the database community due to the need for scalable solutions for big networks. Specific problem formulations in this class ask to measure the probability that two specific nodes are connected (two-terminal reliability \cite{aggarwal1975reliability}), all nodes in the network are pairwise connected (all-terminal reliability \cite{sharafat2009all}), or all nodes in a given subset are pairwise connected (k-terminal reliability \cite{mrsp2007,hardy2007k}). Extending these reliability queries, where source and sink node(s) are specified, the corresponding graph mining problem is to find, for a given probabilistic graph, the set of most reliable k-terminal subgraphs \cite{fastDiscOfRelkTerSub2010}. All these problem definitions have in common that the set of nodes to be reached is predefined, and that there is no degree of freedom in the number of activated edges - thus all nodes are assumed to attempt to communicate to all their neighbors, which we argue can be overly expensive in many applications.

{\bf Reliability Bounds.} Several lower bounds on (two-terminal) reliability have been defined in the context of communication networks \cite{brecht1988lower, bulka1994network, galtier2005algorithms, provan1984computing}. Such bounds could be used in the place of our sampling approach, to estimate the information gain obtained by adding a network edge to the current active set. However, for all these bounds, the computational complexity to obtain these bounds is at least quadratic in the number of network nodes, making these bounds unfeasible for large networks. Very simple but efficient bounds have been presented in \cite{khan2014fast}, such as using the most-probable path between two nodes as a lower bound of their two-terminal reliability. However, the number of possible (non-circular) paths is exponentially large in the number of edges of a graph, such that in practice, even the most probable path will have a negligible probability, thus yielding a useless upper bound. Thus, since none of these probability bounds are sufficiently effective and efficient for practical use, we directly decided to use a sampling approach for parts of the graph where no exact inference is possible.

{\bf Influential Nodes.} Existing work motivated by applications in marketing provide methods to detect \textit{influential} members within a social network. This can help to promote a new product. The task is to detect nodes, i.e., persons, where the chance that the product is recommended to a broad range of connected people is maximized. In \cite{valueCustomer2001}, \cite{viralMarketing2002} a framework is provided which considers the interactions between the persons in a probabilistic model. As the problem of finding the most influential vertices is \textit{NP-hard}, approximation algorithms are used in \cite{kempe2003maximizing} outperforming basic heuristics based on degree centrality and distance centrality which are applied traditionally in social networks. This branch of research has in common that the task is to activate a constrained number of \emph{nodes} to maximize the information flow, whereas our problem definition constrains the number of activated \emph{edges} for a single specified query/sink node.

{\bf Reliable Paths.} In mobile ad hoc networks, the uncertainty of an edge can be interpreted as the connectivity between two nodes. Thus, an important problem in this field is to maximize the probability that two nodes are connected for a constrained budget of edges \cite{routing2007}. In this work, the main difference to our work is that the information flow to a single destination is maximized, rather than the information flow in general. The heuristics \cite{routing2007} cannot be applied directly to our problem, since clearly, maximizing the flow to one node may detriment the flow to another node.

{\bf Bi-connected components.} The \emph{F-tree} that we propose in this work is inspired by the \emph{block-cut tree} \cite{tarjan1972depth,hopcroft1973algorithm,westbrook1992maintaining}. The main difference is that our approach aims at finding cyclic subgraphs, where nodes are bi-connected. For subgraphs having a size of at least three vertices, this problem is equivalent to finding bi-connected subgraphs, which is solved in \cite{tarjan1972depth,hopcroft1973algorithm,westbrook1992maintaining}. Thus, our proposed data structure treats bi-connected subgraphs of size less than three separately, grouping them together as mono-connected components. More importantly, this existing work does not show how to compute, estimate and propagate probabilistic information through the structure, which is the main contribution of this work.
\section{Problem Definition}
\label{sec:fundamentals}
A probabilistic undirected graph is given by $\mathcal{G}=(V,E,W,P)$, where $V$ is a set of vertices, $E\subseteq V\times V$ is a set of edges, $W:V\mapsto \mathbb{R}^+$ is a function that maps each vertex to a positive value representing the information weight of the corresponding vertex and $P:E\mapsto (0,1]$ is a function that maps each edge to its corresponding probability of existing in $\mathcal{G}$. In the following, it is assumed that the existence of different edges are independent from one another. Let us note, that our approach also applies to other models such as the conditional probability model \cite{kneighbors2010}, as long as a computational method for an unbiased drawing of samples of the probabilistic graph is available.

In a probabilistic graph $\mathcal{G}$, the existence of each edge is a random variable. Thus, the topology of $\mathcal{G}$ is a random variable, too. The sample space of this random variable is the set of all \emph{possible graphs}. A \emph{possible graph} $g = (V_g,E_g)$ of a probabilistic graph $\mathcal{G}$ is a deterministic graph which is a possible outcome of the random variables representing the edges of $\mathcal{G}$. The graph $g$ contains a subset of edges of $\mathcal{G}$, i.e. $E_g\subseteq E$. The total number of possible graphs is $2^{|E_{<1}|}$, where $|E_{<1}|$ represents the number of edges $e\in E$ having $P(e)<1$, because for each such edge, we have two cases as to whether or not that edge is present in the graph. We let $\mathcal{W}$ denote the set of all possible graphs. The probability of sampling the graph $g$ from the random variables representing the probabilistic graph $\mathcal{G}$ is given by the following sampling or realization probability $Pr(g)$:
\begin{equation}\label{eq:pw}
Pr(g)=\prod_{e\in E_g}P(e)\cdot \prod_{e\in E\setminus E_g} (1-P(e)).
\end{equation}
Figure \ref{fig:runningExample_a} shows an example of a
probabilistic graph $\mathcal{G}$ and one of its possible realization $g_1$ in
\ref{fig:runningExample_d}.
This probabilistic graph has $2^{10}=1024$ possible worlds. Using Equation
\ref{eq:pw}, the probability of world $g_1$ is given by:
\begin{align*}
Pr(g_1) = &0.6 \cdot 0.5 \cdot 0.8 \cdot 0.4 \cdot 0.4 \cdot 0.5 \cdot
(1-0.1)\cdot \\
& \cdot (1-0.3) \cdot (1-0.4) \cdot (1-0.1) = 0.00653184
\end{align*}
\begin{definition}[Path]\label{def:path}
Let $\mathcal{G}=(V,E,W,P)$ be a probabilistic graph and let $v_0, v_n\in V$ be two nodes such that $v_0 \neq v_n$. An (acyclic) \textbf{path} $path(v_0,v_n) = (v_0, v_1, v_2, \ldots, v_n)$ is a sequence of vertices, such that $(\forall v_i \in path(v_0,v_n)) (v_i\in V)$ and $(\forall v_i \in path(v_0,v_{n-1})) ((v_i, v_{i+1}) \in E)$.
\end{definition}


\begin{definition}[Reachability]\label{def:reach}
The network reachability problem as defined in \cite{jin2011distance,colbourn1987combinatorics} computes the likelihood of the binomial random variable ${\updownarrow}(i,j,\mathcal{G})$ of two nodes $i,j\in V$ being connected in $\mathcal{G}$, formally:
$$
P({\updownarrow}(i,j,\mathcal{G})):=\sum_{g\in \mathcal{W}}\prod_{e\in E_g}P(e)\cdot \prod_{e\in E\setminus E_g} (1-P(e))\cdot {\updownarrow}(i,j,g),
$$
\end{definition}
where ${\updownarrow}(i,j,g)$ is an indicator function that returns one if there exists a path between nodes $i$ and $j$ in the (deterministic) possible graph $g$, and zero otherwise. For a given query node $Q$, our aim is to optimize the information gain, which is defined as
the total weight of nodes reachable from $Q$.

\begin{definition}[\small Expected Information Flow]
Let $Q\in V$ be a node and let $\mathcal{G}=(V,E,W,P)$ be a probabilistic graph, then $\mbox{flow}(Q,\mathcal{G})$ denotes the random variable of the sum of vertex weights of all nodes in $V$ reachable from $Q$, formally:
$$
\mbox{flow}(Q,\mathcal{G}):=\sum_{v\in
V}P({\updownarrow}(Q,v,\mathcal{G}))\cdot W(v).
$$
Due to linearity of expectations, and exploiting that $W(v)$ is deterministic, we can compute the expectation $E(\mbox{flow}(Q,\mathcal{G}))$ of this random variable as $E(\mbox{flow}(Q,\mathcal{G}))=$
\begin{equation}\label{eq:infflow}
E(\sum_{v\in V}P({\updownarrow}(Q,v,\mathcal{G}))\cdot W(v))=\sum_{v\in
V}E(P({\updownarrow}(Q,v,\mathcal{G})))\cdot W(v)
\end{equation}
\end{definition}
\noindent Given our definition of Expected Information Flow in Equation~\ref{eq:infflow}, we can now state our formal problem definition of optimizing the expected information flow of a probabilistic graph $\mathcal{G}$ for a constrained budget of edges.
\begin{definition}[Maximum Expected Information Flow]\label{def:maxinfflow} Let $\mathcal{G}=(V,E,W,P)$ be a probabilistic graph, let $Q\in V$ be a query node and let $k$ be a non-negative integer. The Maximum Expected Information Flow
\begin{equation}
\mbox{MaxFlow}(\mathcal{G},Q,k)=
$$\vspace{-0.6cm}
$$
argmax_{G=(V,E'\subseteq E,W,P),|E^\prime|\leq k}E(\mbox{flow}(Q,G)),
\end{equation}
is the subgraph of $\mathcal{G}$ maximizing the information flow towards $Q$
constrained to having at most $k$ edges.
\end{definition}
Computing $\mbox{MaxFlow}(\mathcal{G},Q,k)$ efficiently requires to overcome two NP-hard subproblems. First, the computation of the expected information flow $E(\mbox{flow}(Q,G))$ to vertex $Q$ for a given probabilistic graph $\mathcal{G}$ is NP-hard as shown in \cite{colbourn1987combinatorics}. In addition, the problem of selecting the optimal set of $k$ vertices to maximize the information flow $\mbox{MaxFlow}(\mathcal{G},Q,k)$ is a NP-hard problem in itself, as shown in the following.
\begin{theorem}\label{the:hard}
Even if the Expected Information Flow $\mbox{flow}(Q,\mathcal{G})$ to a vertex $Q$ can be computed in $O(1)$ for any probabilistic graph $\mathcal{G}$, the problem of finding $\mbox{MaxFlow}(\mathcal{G},Q,k)$ is still NP-hard.
\end{theorem}
\begin{proof}
In this proof, we will show that a special case of computing
$\mbox{MaxFlow}(\mathcal{G},Q,k)$ is \emph{NP-complete}, thus implying that our general problem is \emph{NP-hard}.
We reduce the \emph{0-1 knapsack problem} to the problem of computing
$\mbox{MaxFlow}(\mathcal{G}, Q, k)$. Thus, assume a 0-1 knapsack problem: Given
a capacity integer $W$ and given a set $\{i_1,...,i_n\}$ of $n$ items each
having an integer weight $w_i$ and an integer value $v_i$. The \emph{0-1 knapsack} problem is to find the optimal vector $x=(x_1,...,x_n)\in \{0,1\}^n$ such that $\sum_{i=1}^n v_i\cdot x_i$, subject to $\sum_{i=1}^n w_i \cdot x_i \leq W$. This problem is known to be \emph{NP-complete} \cite{kellerer2004introduction}. We reduce this problem to the problem of computing $\mbox{MaxFlow}(\mathcal{G}, Q, k)$ as follows. Let $\mathcal{G}=(V,E,G,P)$ be a probabilistic graph such that $Q$ is connected to $n$ nodes $\{V_1,...,V_n\}$ (one node for each item of the knapsack problem). Each node $V_i$ is connected to a chain of $w_i-1$ nodes $\{V_i^1,...,V_i^{w_i-1}\}$. All edges have a probability of one, i.e., $P(v\in V)=1$. The information of a node is set to $w_i$ if it is the (only) leaf node $v_i^{w_i-1}$ of the branch of $\mathcal{G}$ connected to $V_i$ and zero otherwise. Finally, set $k=W$. Then, the solution of the \emph{0-1 knapsack} problem can be derived from the constructed $\mbox{MaxFlow}(\mathcal{G}, Q, k)$ problem by selecting all items $n_i$ such that the corresponding node $v_i^{w_i-1}$ is connected to $Q$. Thus, if we can solve the $\mbox{MaxFlow}(\mathcal{G}, Q, k)$ problem in polynomial time, then we can solve the \emph{0-1 knapsack problem} in polynomial time: A contradiction assuming $P\neq NP$.
\begin{figure}[h]
  \begin{center}
    \includegraphics[width=14pc]{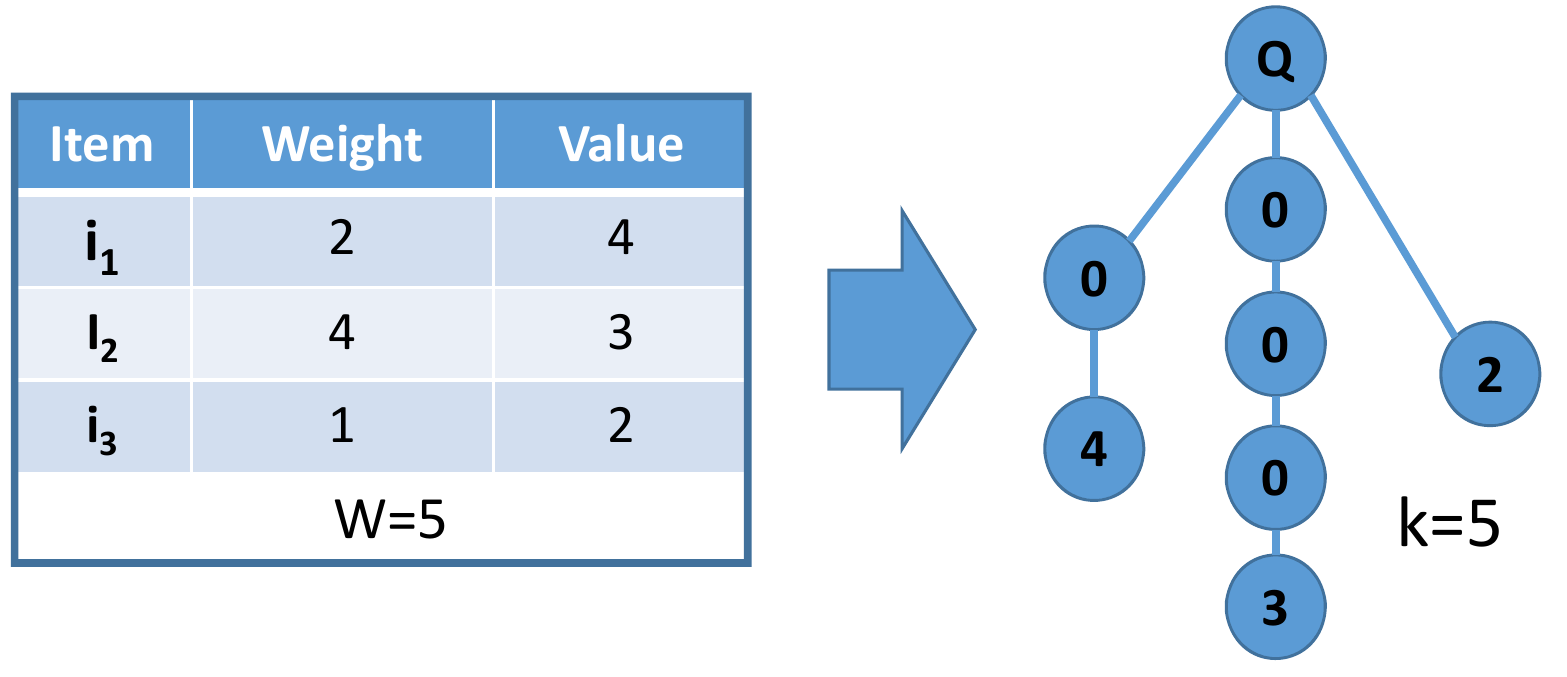}	
  \end{center}
	\caption{\label{fig:knapsack}Example of the Knapsack Reduction
	of Theorem \ref{the:hard}}
	\vspace{-0.4cm}
\end{figure}
\end{proof}

%
%

\section{Roadmap}
\label{sec:roadmap} To compute $\mbox{MaxFlow}(\mathcal{G},Q,k)$, we first need an efficient solution to approximate the reachability probability $E({\updownarrow}(Q,v,\mathcal{G}))$ from $Q$ to/from a single node $v$. Since this problem can be shown to be \emph{\#P-hard}, Section \ref{sec:ct} presents an approximation technique which exploits stochastic independencies between branches of a spanning tree of subgraph $G$ rooted at $Q$.  This technique allows to aggregate independent subgraphs of $G$ efficiently, while exploiting a sampling solution for components of the graph $\mbox{MaxFlow}(\mathcal{G},Q,k)$ that contains cycles.

Once we can efficiently approximate the flow $E({\updownarrow}(Q,v,\mathcal{G}))$ from $Q$ to each node $v\in V$, we next tackle the problem of efficiently finding a subgraph $\mbox{MaxFlow}(\mathcal{G},Q,k)$ that yields a near-optimal expected information flow given a budget of $k$ edges in Section \ref{sec:algorithms}. Due to the theoretic result of Theorem \ref{the:hard}, we propose heuristics to choose $k$ edges from $\mathcal{G}$. Finally, our experiments in Section \ref{sec:evaluation} support our theoretical intuition that our solutions for the two aforementioned subproblems synergize: An optimal subgraph will choose a budget of $k$ edges in a tree-like fashion, to reach large parts of the probabilistic graph. At the same time, our solutions exploit tree-like subgraphs for efficient probability computation.\vspace{-0.25cm}
\section{Expected Flow Estimation}
\label{sec:reachability} 
In this section we estimate the expected information flow of a given subgraph $G
\subseteq \mathcal{G}$. Following Equation \ref{eq:infflow}, the reachability probability $P({\updownarrow}(Q,v,\mathcal{G}))$ between $Q$ and a node $v$ can
be used to compute the total expected information flow
$E(\mbox{flow}(Q,\mathcal{G}))$. This problem of computing the reachability
probability between two nodes has been shown to be \emph{$\#P$-hard}
\cite{routing2007, colbourn1987combinatorics} and sampling solutions
have been proposed to approximate it \cite{Sampling06,emrich2012exploration}. In this section, we will present our
solution to identify subgraphs of $\mathcal{G}$ for which we can compute the
information analytically and efficiently, such that expensive numeric sampling
only has to be applied to small subgraphs.
We first introduce the concept of Monte-Carlo sampling of a subgraph.
\subsection{Traditional Monte-Carlo Sampling}
\begin{lemma}\label{lem:sampling}
Let $\mathcal{G}=(V,E,W,P)$ be an uncertain graph
and let $\mathcal{S}$ be a set of sample worlds drawn randomly and unbiased from
the set $\mathcal{W}$ of possible graphs of $\mathcal{G}$. Then the average information flow in samples in $\mathcal S$
\begin{equation}\label{eq:infflowsampling}
 \frac{1}{|\mathcal{S}|}\sum_{g\in \mathcal{S}}\mbox{flow}(Q,g) = \frac{1}{|\mathcal{S}|}\cdot \sum_{g\in
\mathcal{S}} \sum_{v} {\updownarrow}(Q,v,g)\cdot W(v)
\end{equation}
is an unbiased estimator of the expected information flow
$E(\mbox{flow}(Q,\mathcal{G}))$, where ${\updownarrow}(Q,v,g)$ is an indicator
function that returns one if there exists a path between nodes $Q$ and $v$ in the (deterministic) sample graph $g$, and zero otherwise.
\end{lemma}
\begin{proof}
For $\mu$ to be an unbiased estimator of $E(\mbox{flow}(Q,\mathcal{G}))$, we have to show that $E(\mu)=E(\mbox{flow}(Q,\mathcal{G}))$.
Substituting $\mu$ yields $E(\mu)=E(\frac{1}{|\mathcal{S}|}\sum_{g\in
\mathcal{S}}\mbox{flow}(Q,g))$. Due to linearity of expectations, this is equal
to $\frac{1}{|\mathcal{S}|}\sum_{g\in \mathcal{S}}E(\mbox{flow}(Q,g))$. The sum
over $|\mathcal{S}|$ identical values can be replaced by a factor of
$|\mathcal{S}|$.
Reducing this factor yields $E(\mbox{flow}(Q,g\in\mathcal{S}))$.
Following the assumption of unbiased sampling $\mathcal{S}$ from the set $\mathcal{W}$ of possible worlds, the expected information flow $E(\mbox{flow}(Q,g))$ of a sample possible world $g\in\mathcal{S}$ is equal to the expected information flow $E(\mbox{flow}(Q,\mathcal{G}))$.
\end{proof}

Naive sampling of the whole graph $\mathcal{G}$ has disadvantages:
First, this approach requires to compute reachability queries on a set of
possibly large sampled graphs. Second, a rather large approximation error is
incurred. We will approach these drawbacks by first describing how
non-cyclic subgraphs, i.e. trees, can be processed in order to compute
the information flow exactly and efficiently without sampling.
For cyclic subgraphs, we show how sampled information flows can be used to compute the information flow in the full graph.

\subsection{Mono-Connected vs. Bi-Connected graphs}
\label{subsec:treeStructures}
The main observation that will be exploited in our work is the following: if there exists only one
possible path between two vertices, then we can compute their reachability
probability efficiently.
\begin{definition}[Mono-Connected Nodes]\label{def:mono}
Let $\mathcal{G}=(V,E,W,P)$ be a probabilistic graph and let $A,B\in V$. If
$\mbox{path}(A,B)=(A=v_0,v_1,...,v_{k-1},v_k=B)$ is \emph{the only} path
between $A$ and $B$, i.e., there exists no other path $p\in V\times V\times
V^*$ that satisfies Definition \ref{def:path}, then we denote $A$ and $B$ as \emph{mono-connected}.
\end{definition}
In the following, when the query vertex $Q$ is clear from the context, we call a vertex $A$ mono-connected if it is mono-connected to the query vertex $Q$.
\begin{lemma}\label{lem:path}
If two vertices $A$ and $B$ are mono-connected in a probabilistic graph $\mathcal{G}$, then the reachability
probability between $A$ and $B$ is equal to the product of the edge
probabilities included in $\mbox{path}(A,B)$, i.e., $$
{\updownarrow}(A,B,\mathcal{G})=\prod_{i=0}^{k-1}P((v_i,v_{i+1})) \text{ with }
v_i \in path(A,B)$$
\end{lemma}

\begin{proof}
Following possible world semantics as defined in Definition \ref{def:reach},
the reachability probability ${\updownarrow}(A,B,\mathcal{G})$ is the sum of probabilities of all
possible worlds where $B$ is connected to $A$.
We show that $A$ and $B$ are connected in a possible graph $g$ iff
all $k-1$ edges $e_i=(v_i,v_{i+1})$ with $v_i, v_{i+1}\in \mbox{path}(A,B)$ exist. \\
$\Rightarrow$: By contradiction: Let $A$ and $B$ be connected in $g$, and let any edge on $\mbox{path}(A,B)$ be missing. Then there must exist a path $\mbox{path}^{prime}(A,B)\neq \mbox{path}(A,B)$ which contradicts the assumption that $A$ and $B$ are mono-connected.\\
$\Leftarrow$: If all edges on $\mbox{path}(A,B)$ exist, then $B$ is connected to
$A$ following the assumption that $\mbox{path}(A,B)$ is a path from $A$ to $B$.

Due to our assumption of independent edges, the probability that all edges in
$\mbox{path}(A,B)$ exist is given by $\prod_{i=0}^{k-1}P((v_i,v_{i+1}))$.
\end{proof}

\begin{definition}[Mono-Connected Graph]\label{def:mono-graph}
A probabilistic graph $\mathcal{G}=(V,E,W,P)$ is called \textit{mono-connected}, iff all pairs of vertices in $V$ are mono-connected.
\end{definition}

Next, we generalize Lemma \ref{lem:path} to whole subgraphs, such that a specified vertex $Q$ in that subgraph has a unique path to all other vertices in the subgraph. Using Lemma \ref{lem:path}, we constitute the following theorem that will be
exploited in the remainder of this work.
\begin{theorem}\label{the:exp}
Let $\mathcal{G}=(V,E,G,P)$ be a probabilistic graph, let $Q\in V$ be a node. If $\mathcal{G}$ is mono-connected, then $E(\mbox{flow}(Q,\mathcal{G}))$ can be
computed efficiently.
\end{theorem}

\begin{proof}
$E(\mbox{flow}(Q,\mathcal{G}))$ is the sum of reachability probabilities of all nodes, according to Equation \ref{eq:infflow}.  If $\mathcal{G}$ is connected and non-cyclic, we can guarantee that each node has exactly one path to $Q$, and thus, is mono-connected. Thus, Lemma \ref{lem:path} is applicable to compute the reachability probability between $Q$ and each node $v\in V$. Due to linearity of expectations, i.e., $E(X+Y)=E(X)+E(Y)$ for random variables $X$ and $Y$, we can aggregate individual reachability expectations, yielding $E(\mbox{flow}(Q,\mathcal{G}))$.
\end{proof}
\noindent Analogously to Definition \ref{def:mono}, we define bi-connected nodes.
\begin{definition}[Bi-Connected Nodes]\label{def:bi}
Let $\mathcal{G}=(V,E,W,P)$ be a probabilistic graph and let $A,B\in V$. If
there exists (at least) two paths $\mbox{path}_1(A,B)$ and $\mbox{path}_2(A,B)$,
such that $\mbox{path}_1(A,B)\neq \mbox{path}_2(A,B)$, then we denote $A$ and
$B$ as \emph{bi-connected}.
\end{definition}
\vspace{.1cm}
\begin{definition}[Bi-Connected Graph]\label{def:bi-graph}
A bi-connected graph \cite{tarjan1972depth,hopcroft1973algorithm} is a
connected probabilistic graph $\mathcal{G}=(V,E,W,P)$ such that removal of
any vertex $A\in V$ will still yield a connected probabilistic graph.
\end{definition}
\begin{lemma}\label{lem:bi}
In a bi-connected graph $\mathcal{G}$ of size $|V|\geq 3$, all pairs of vertices are bi-connected following Definition \ref{def:bi}.
\end{lemma}
\begin{proof}
By contradiction, let $A,B$ be two nodes in $\mathcal{G}$ that are mono-connected. Let $\mbox{path}(A,B)$ be the only path between them.
\\Case 1: $\mbox{path}(A,B)=(A,B)$ contains no other vertices: Since
$\mathcal{G}$ is bi-connected, removal of vertex $A$ yields a graph where $B$ and $C$ are connected by some path $\mbox{path}(C,B)$. At the same time, removal of vertex $B$ yields a graph where $A$ and $C$ are connected by some path $\mbox{path}(A,C)$. Thus, the concatenation of these paths yields an alternative path between $A$ and $B$, contradicting the assumption that $(A,B)$ are mono-connected by path $(A,B)$. \\Case 2: $\mbox{path}(A,B)=(A,C_1,...,C_n,B)$ contains other vertices. Let $C_1$ be such a vertex. Since $\mathcal{G}$ is bi-connected, removal of vertex $C_1$ yields a graph where $A$ and $B$ are still connected, contradicting the assumption that $A$ and $B$ are mono-connected by $\mbox{path}(A,B)$ only.
\end{proof}

%
The information flow within a bi-connected graph can not be computed efficiently using Theorem \ref{the:exp}, as the flow between any two nodes $A$ to $B$ is shared by more than one path. In the next section, we propose techniques to substitute bi-connected subgraphs by super-nodes, for which we can estimate the information flow using Monte-Carlo sampling exploiting Lemma \ref{lem:sampling}. By substituting the bi-connected subgraphs by super-nodes for which we apply sampling and memoize the sampling information for these super-nodes, we yield a mono-connected graph that uses the substituted super-nodes. This approach maximizes the partitions of the graph for which expensive Monte-Carlo estimation can be replaced using Theorem \ref{the:exp}.

The next section will show how to achieve this goal, by employing a \emph{F-tree} of the graph. This data structure borrowed from graph theory partitions the graph into bi-connected components (a.k.a. ``blocks'') generated by bi-connected subgraphs, and identifies vertices of the graph as \emph{articulation vertices} to connect two bi-connected components. We exploit these articulation vertices, by having them represent all the information flow that is estimated to flow to them from their corresponding bi-connected component. 
\vspace{-0.25cm}
\subsection{Flow tree}
\label{sec:ct}
\begin{figure}
  \subfigure[Example Graph]{
      \includegraphics[width=0.6\columnwidth]{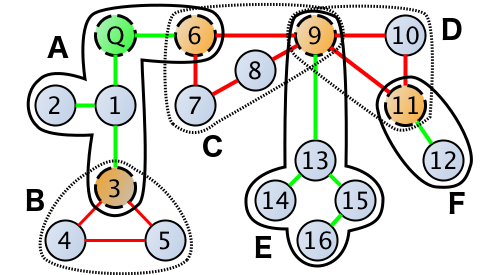}
    \label{fig:ct_example_graph}
    }
  \centering
  \subfigure[F-tree~representation]{
      \includegraphics[width=.3\columnwidth]{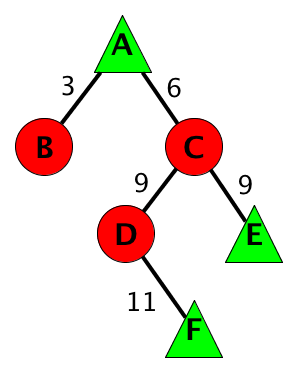}
    \label{fig:ct_example_ct}
    }\vspace{-0.3cm}
	\caption{Running example graph with corresponding \emph{F-tree}}\vspace{-0.5cm}
	\label{fig:ct_example}
\end{figure}
In this section, we propose to adapt the block-cut tree
\cite{tarjan1972depth,hopcroft1973algorithm,westbrook1992maintaining} to
partition a graph into independent bi-connected components. Instead of sampling
the whole uncertain graph, the purpose of this index structure is to exploit
Theorem \ref{the:exp} for mono-connected components, and to apply local
Monte-Carlo within bi-connected components only.
Our employed \emph{Flow tree} (F-tree) memoizes the information flow at each
node.
Before we show how to utilize the \emph{F-tree} for efficient information flow
computation, we first give a formal definition.
\begin{definition}[Flow tree]\label{def:CT}
Let $\mathcal{G}=(V,E,W,P)$ be a probabilistic graph and let $Q\in V$ be a vertex for which the expected information flow is computed.
A \emph{Flow tree} (\emph{F-tree}) is a tree structure defined as follows.

\noindent {\bf 1)} each component of the \emph{F-tree} is a \emph{connected
subgraph} of $\mathcal{G}$. A component can be \emph{mono-connected} or
\emph{bi-connected}.\\
\noindent {\bf 2)} a mono-connected component $MC=(MC.V\subseteq V,MC.AV\in V)$ is a
set of vertices $MC.V\cup MC.AV$ that form a mono-connected subgraph (c.f. Definition \ref{def:mono-graph}) in $\mathcal{G}$. The vertex
$MC.AV$ is called \emph{articulation vertex}. Intuitively, a mono-connected
components represents a tree-like structure rooted in $MC.AV$. Using Theorem
\ref{the:exp}, we can efficiently compute the information flow from all vertices $MC.V$ to $MC.AV$.\\
\noindent {\bf 3)} a bi-connected component $BC=(BC.V,BC.P(v)$, $BC.AV)$ is a set of vertices
$BC.V\cup BC.AV$ of size greater than two that form a bi-connected subgraph
$\mathcal{G}^\prime$ in $\mathcal{G}$ according to Definition \ref{def:bi-graph}. Intuitively, a bi-connected component
represents a subgraph describing a cycle. In this case, we can estimate the likelihood of being connected to the articulation vertex $BC.AV$ using Monte-Carlo sampling in Lemma
\ref{lem:sampling}.
The function $BC.P(v):BC.V\mapsto [0,1]$ maps each vertex $v\in BC.V$ to the
estimated reachability probability $reach(v,BC.AV)$ of $v$ being connected to
$BC.AV$ in $\mathcal{G}$.\\
\noindent {\bf 4)} For each pair of (mono- or bi-connected) components $(C_1,C_2)$, it holds
that the intersection $C_1.V\cap C_2.V=\emptyset$ of vertices is empty. Thus,
each vertex in $V$ is mapped to at most one component's vertex set.\\
\noindent {\bf 5)} Two different components may share the same articulation vertex, and the
articulation vertex of one component may be in the vertex set of another
component.\\
\noindent {\bf 6)} The articulation vertex of the root of a \emph{F-tree} is
$Q \in V$.
\end{definition}
Intuitively speaking, a component is a set of vertices together with an
articulation vertex that all information must flow through in order to reach
$Q$. By our iterative construction algorithm presented in Section \ref{subsec:ct_update}, each component is guaranteed to have such an articulation vertex,
guiding the direction to vertex $Q$. The idea of the \emph{F-tree} is to use
components as virtual nodes, such that all actual vertices of a component
send their information to their articulation vertex. Then the articulation
vertex forwards all information to the next component, until the root of the
tree is reached where all information is sent to articulation vertex $Q$.
\begin{example}
As an example for a \emph{F-tree}, consider Figure \ref{fig:ct_example_graph},
showing a probabilistic graph. For brevity, assume that each edge $e\in E$ has
an existential probability of $p(e)=0.5$ and that all vertices $v\in V$ have an information weight corresponding to their id, e.g. vertex $6$ has a weight of six. A corresponding \emph{F-tree} is shown in Figure \ref{fig:ct_example_ct}.
A mono-connected component is given by $A=(\{1,2,3,6\},Q)$. For this component,
we can exploit Theorem \ref{the:exp} to analytically compute the flow of
information from any vertex in $\{1,2,3,6\}$ to articulation vertex $Q$: vertices $3$ and $6$ are connected to $Q$ with probability $0.5$. Thus, these nodes contributed an expected information flow of $3\cdot 0.5=1.5$ and $6\cdot 0.5=3$ respectively. Vertices $2$ and $3$ are connected to $Q$ with a probability of $0.5\cdot 0.5=0.25$, respectively, following Lemma \ref{lem:path}. Thus, these nodes contribute an expected information of $2\cdot 0.25=0.5$ and $3\cdot 0.25=0.75$. Following Theorem \ref{the:exp}, we can aggregate these probabilities to obtain the expected information flow from vertices $\{1,2,3,6\}$ to articulation vertex $Q$ as $5.75$.

A bi-connected component is defined by $B=(\{4,5\},3)$, representing a sub-graph
having a cycle. Having a cycle, we cannot exploit Theorem  \ref{the:exp} to
compute the flow of a vertex in $\{4,5\}$ to vertex $3$. But we can sample the
subgraph spanned by vertices in $\{3,4,5\}$ to estimate probabilities that
vertices $\{4,5\}$ are connected to articulation vertex $3$ using Lemma
\ref{lem:sampling}. With sufficient samples, this will yield a probability of
around $0.375$ for both vertices to be reached. Again using Theorem
\ref{the:exp}, we compute an information flow of $0.375\cdot 4+0.375\cdot 5=3.375$ to articulation vertex
$3$. Given this expected flow, we can use the mono-connected component $A$ to
compute the expected information analytically that is further propagated from
the articulation vertex $3$ of component $B$ to the articulation vertex $Q$ of
$A$. As the articulation vertex of component $B$ is in the vertex set of
component $A$, component $B$ is a child of component $A$ in Figure \ref{fig:ct_example_ct} since $B$ propagates its information to $A$.
As we have already computed above, the probability of vertex $3$ to be connected to its articulation vertex $Q$ is $0.25$, yielding an information flow worth $3.375\cdot 0.25=0.84375$ units flowing from vertices $\{4,5\}$ to $Q$. Again, exploiting Theorem \ref{the:exp}, we can aggregate this to a total flow of $5.75+0.84375=6.59375$ from vertices \{1,2,3,4,5,6\} to $Q$.

Another bi-connected component is $C=(\{7,8,9\},6)$, for which we can estimate
the information flow from vertices $7$, $8$, and $9$ to articulation vertex $6$
numerically using Monte-Carlo sampling. Since vertex $6$ is in $A$, component $C$ is a child
of $A$. We find another bi-connected component $D=(\{10,11\},9)$, and two more
mono-connected components $E=(\{13,14,15,16\},9)$ and $F=(\{12\},11)$.
\end{example}
In this example, the structure of the \emph{F-tree} allows us to compute or approximate the expected information flow to $Q$ from each vertex. For this purpose, only three small components $B$ and $C$ and $D$ need to be sampled. This is a vast reduction of sampling space compared to a naive Monte-Carlo approach that samples the full graph: rather than sampling a single random variable having $2^{|E|}=2^{19}=524288$ possible worlds, we only need to sample three random variables corresponding to the bi-connected components $B$, $C$ and $D$ having $2^{3}=8$, $2^{4}=16$, and $2^3=8$ possible worlds, respectively. Clearly, this approach reduces the number of edges (marked in red in Figure \ref{fig:ct_example_graph}) that need to be sampled in each sampling iteration. More importantly, our experiments show that this approach of sampling component independently vastly decreases the variance of the total information flow, thus yielding a more precise estimation at the same number of samples.

Having defined syntax and semantics of the \emph{F-tree}, the next section shows
how to maintain the structure of a \emph{F-tree} when additional edges are
selected. It is important to note that we do not intend to insert all edges of a
probabilistic graph $\mathcal{G}$ into the \emph{F-tree}. Rather, we only add
the edges that are selected to compute the maximum flow
$\mbox{MaxFlow}(\mathcal{G},Q,k)$ given a constrained budget of $k$ edges. Thus,
even in a case where all vertices a bi-connected, such as in the initial example
in Figure \ref{fig:runningExample_a}, we note, supported by our experimental
evaluation, that an optimal selection of edges prefers a spanning-tree-like
topology, which synergizes well with our \emph{F-tree}. The next section shows
how to build the structure of the \emph{F-tree} iteratively by adding edges to
an initially empty graph.

The next subsection proposes an algorithm, to update a \emph{F-tree} when a new
edge is selected, starting at a trivial \emph{F-tree} that contains only one
component $(\emptyset,Q)$. Using this edge-insertion algorithm, we will show how
to choose promising edges to be inserted to maximize the expected information
flow. The selection of the edges of the \emph{F-tree} will
be shown in section \ref{sec:algorithms}.\vspace{-0.25cm}
\subsection{Insertion of Edges into a F-tree}
\label{subsec:ct_update}
%
Following Definition \ref{def:CT} of a \emph{F-tree}, each vertex
$v\in\mathcal{G}$ is assigned to either a single mono-connected component (noted by a flag $v.isMC$ in the algorithm below), a single bi-connected component (noted by $v.isBC$), or to no component, and thus disconnected from $Q$, noted by $v.isNew$. To insert a new edge $(v_{src},v_{dest})$, our edge-insertion algorithm derived in this section differs between these cases as follows:
\\\textbf{\bf Case I)} {$v_{src}.isNew$ and $v_{dest}.isNew$}: We omit this
case, as our edge selection algorithms presented in Section \ref{sec:algorithms} always ensure a single connected component and initially the \emph{F-tree} contains only vertex $Q$.
\\\textbf{\bf Case II)} {$v_{src}.isNew$ exclusive-or $v_{dest}.isNew$}: Due to
considering undirected edges, we assume without loss of generality that
$v_{dest}.isNew$. Thus $v_{src}$ is already connected to \emph{F-tree}.

{\bf Case IIa)}: $v_{src}.isMC$: In this case, a new dead end is added
to the mono-connected structure $MC_{src}$ which is guaranteed to remain mono-connected. We add $v_{dest}$ to $MC_{src}.V$.

{\bf Case IIb)}: $v_{src}.isBC$: In this case, a new dead end is added to
the bi-connected structure $BC_{src}$. This dead end becomes a new
mono-connected component $MC=(\{v_{dest}\},v_{src})$. Intuitively speaking, we
know that vertex $v_{dest}$ has no other choice but propagating its information
to $v_{src}$. Thus, $v_{src}$ becomes the articulation vertex of $MC$. The
bi-connected component $BC_{src}$ adds the new mono-connected component $MC$ to
its list of children.
\\\textbf{\bf Case III)} $v_{src}$ and $v_{dest}$ belong to the same component,
i.e.
$C_{src} = C_{dest}$

{\bf Case IIIa)} This component is a bi-connected component $BC$:
  Adding a new edge between $v_{src}$ and $v_{dest}$ within component $BC$ may change
  the reachability $BC.P(v)$ of each vertex $v\in BC.V$ to reach their
  articulation vertex $BC.AV$. Therefore, $BC$ needs to be re-sampled to
  numerically estimate the reachability probability function $P(v)$ for each
  $v\in BC.V$.

{\bf Case IIIb)}: This component is a mono-connected component $MC$: In
  this case, a new cycle is created within a mono-connected component, thus some
  vertices within $MC$ may become bi-connected. We need to (i) identify the set
  of vertices affected by this cycle, (ii) split these vertices into a new
  bi-connected component, and (iii) handle the set of vertices that have been
  disconnected from $MC$ by the new cycle. These three steps are performed by
  the \emph{splitTree}$(MC,v_{src},v_{dest})$ function as follows: (i) We start
  by identifying the new cycle as follows: Compare the (unique) paths of
  $v_{src}$ and $v_{dest}$ to $MC.AV$, and find the first vertex $v_{\wedge}$
  that appears in both paths. Now we know that the new cycle is decribed by
  $path(v_\wedge,v_{src}), path(v_{dest},v_\wedge)$ and the new edge between
  $v_{src}$ and $v_{dest}$.
  (ii) All of these vertices are added to a bi-connected component $BC=(path(v_\wedge, v_{src})\cup
  path(v_{dest},v_\wedge)\setminus v_\wedge,P(v), v_\wedge)$ using $v_\wedge$ as
  their articulation vertex. All vertices in $MC$ having $v_\wedge$ (except
  $v_\wedge$ itself) on their path are removed from $MC$. The probability mass
  function $P(v)$ is estimated by sampling the subgraph of vertices in $BC.V$.
  (iii) Finally, orphans of $MC$ that have been split off from $MC$ due to the
  creation of $BC$ need to be collected into new mono-connected components. Such
  orphans having a vertex of the cycle $BC$ on their path to $MC.AV$ will be
  grouped by these vertices: For each $v_i\in BC.V$, let
  $orphan_i$ denote the set of orphans separated by $v_i$ (separated means $v_i$ being the
  first vertex in $BC.V$ on the path to $MC.AV$). For each such group, we create
  a new mono-connected component $MC_i=(orphan_i,v_i)$. All these new
  mono-connected components with $v_i \in BC.V$ become children of $BC$. If
  $MC.V$ is now empty, thus all vertices of $MC$ have been reassigned to other
  components, then $MC$ is deleted and $BC$ will be appended to the list of
  children of the component $C$ where $BC.AV=v_\wedge \in C.V$. In case of
  $MC.V$ being not empty, we are left over with a mono-connected component $MC$ with
  $v_\wedge \in MC.V$. The new bi-connected component $BC$ becomes a child of
  $MC$.
\\\textbf{\bf Case IV)} $v_{src}$ and $v_{dest}$ belong to different components
$C_{src} \neq C_{dest}$.
Since the \emph{F-tree} is a tree-structure itself, we can identify the lowest
common ancestor $C_{anc}$ of $C_{src}$ and $C_{dest}$.
The insertion of edge $(v_{src},v_{dest})$ has incurred a new cycle $\bigcirc$
going from $C_{anc}$ to $C_{src}$, then to $C_{dest}$ via the new edge, and then
back to $C_{anc}$. This cycle may cross mono-connected and bi-connected
components, which all have to be adjusted to account for the new cycle. We need
to identify all vertices involved to create a new cyclic, thus bi-connected,
component for $\bigcirc$, and we need to identify which parts remain
mono-connected. In the following cases, we adjust all components involved in
$\bigcirc$ iteratively. First, we initialize $\bigcirc=(\emptyset,P,v_{anc})$,
where $v_{anc}$ is the vertex within $C_{anc}$ where the cycle meets if $C_{anc}$ is a mono-connected component, and $C_{anc}.AV$ otherwise.
Let $C$ denote the component that is currently adjusted:

{\bf Case IVa)} $C=C_{anc}$: In this case, the new cycle may enter
$C_{anc}$ from two different articulation vertices. In this case, we apply \emph{Case III}, treating these two vertices as $v_{src}$ and $v_{dest}$, as these two vertices have become connected transitively via the big cycle $\bigcirc$.

{\bf Case IVb)} $C$ is a bi-connected component: In this case $C$
becomes absorbed by the new cyclic component $\bigcirc$, thus $\bigcirc.V=\bigcirc.V \cup C.V$, and $\bigcirc$ inherits all children from $C$. The rational is that all vertices within $C$ are able to access the new cycle.

{\bf Case IVc)} $C$ is a mono-connected component: In this case, one
path in $C$ from one vertex $v$ to $C.AV$ is now involved in a cycle. All vertices
involved in $path(v, C.AV)$ are added to $\bigcirc.V$ and removed from $C$. The
operation \emph{splitTree$(C,v,C.AV)$} is called to create new mono-connected
components that have been split off from $C$ and become connected to $\bigcirc$
via their individual articulation vertices.

In the following, we use the graph of Figure \ref{fig:ct_example_graph} and its
corresponding \emph{F}-tree representation of Figure \ref{fig:ct_example_graph}
to insert additional edges and to illustrate the interesting cases of the
insertion algorithm of Section \ref{subsec:ct_update}.

\begin{figure*}
  \hfil
  \subfigure[\textbf{Case IIb}: Insertion of edge $a$.]{
    \parbox{0.225\textwidth}{
      \centering
      \includegraphics[width=0.225\textwidth]{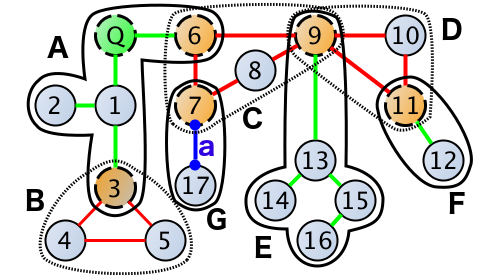}
      \includegraphics[width=0.17\textwidth,
      height=2cm]{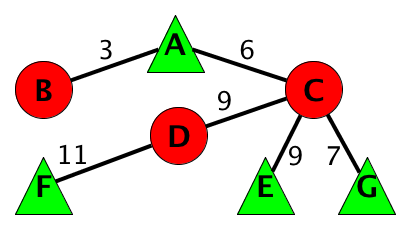} }
    \label{fig:ct_example_graph_a}
    }
  \hfil
  \subfigure[\textbf{Case IIIa}: Insertion of edge $b$.]{
    \parbox{0.225\textwidth}{
      \centering
      \includegraphics[width=0.225\textwidth]{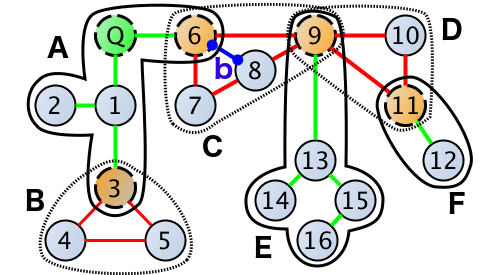}
      \includegraphics[width=0.17\textwidth,
      height=2.25cm]{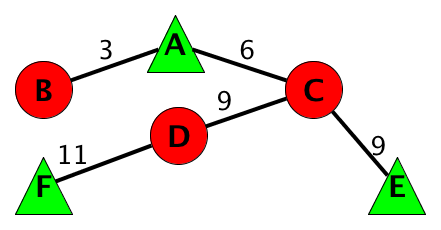} }
    \label{fig:ct_example_graph_b}
    }
    \hfil
  \subfigure[\textbf{Case IIIb}: Insertion of edge~$c$]{
    \parbox{0.225\textwidth}{
      \centering
      \includegraphics[width=0.225\textwidth]{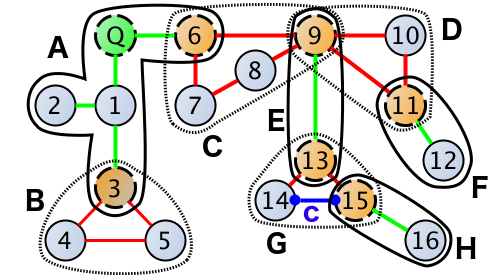}
      \includegraphics[width=0.17\textwidth,
      height=2.25cm]{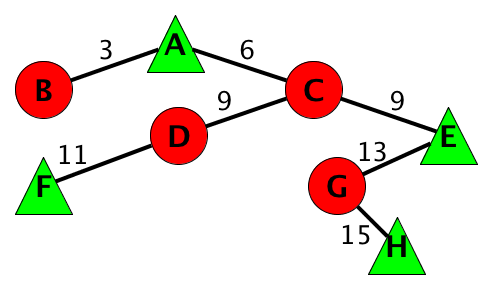} }
    \label{fig:ct_example_graph_c}
    }
  \hfil
  \subfigure[\textbf{Case IVa-c}: Insertion of edge $d$]{
    \parbox{0.225\textwidth}{
      \centering
      \includegraphics[width=0.225\textwidth]{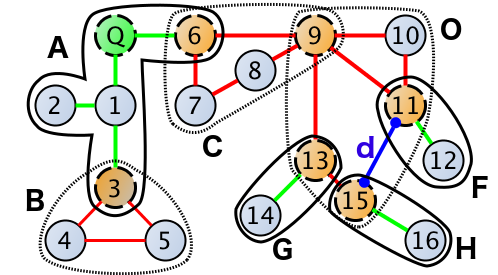}
      \includegraphics[width=0.15\textwidth,
      height=2.25cm]{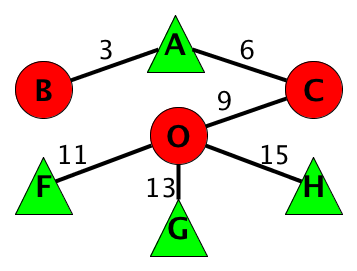} }
    \label{fig:ct_example_graph_d}\vspace{-0.3cm}
    }\vspace{-0.3cm}
    \caption{\vspace{-0.3cm}Examples of edge insertions and \emph{F}-tree update
    cases using the running example of Figure
    \ref{fig:ct_example_graph}.}
\end{figure*}
\subsection{Insertion Examples}
In the following, we use the graph of Figure \ref{fig:ct_example_graph} and its
corresponding \emph{F}-tree (FT) representation of Figure
\ref{fig:ct_example_ct} to insert additional edges and to illustrate the interesting cases of the insertion
algorithm of Section \ref{subsec:ct_update}.

We start by an example for \textbf{\underline{Case II}} in Figure \ref{fig:ct_example_graph_a}. Here, we insert a new edge $a=(7,17)$, thus connecting a new vertex $17$ to the \emph{FT}. Since vertex $7$ belongs to the bi-connected component $BC$, we apply \textbf{\underline{Case IIb}}. A new mono-connected component $G=(\{17\},7)$ is created and added to the children of $BC$.

In Figure \ref{fig:ct_example_graph_b}, we insert a new edge $b=(6,8)$ instead. In this case, the two connected vertices are already part of the \emph{FT}, thus Case II does not apply. We find that both vertices belong to the same component $C$. Thus, \textbf{\underline{Case III}} is used and more specifically, since component $C$ is a bi-connected component $BC$, \textbf{\underline{Case IIIa}} is applied. In this case, no components need to be changed, but the probability function $BC.P(v)$ has to re-approximated, as the probabilities of nodes $7$, $8$ and $9$ will have increased probability of being connected to articulation vertex $6$, due to the existence of new paths arising by inserting edge $b$.

Next, in Figure \ref{fig:ct_example_graph_c}, an edge is inserted between vertices $14$ and $15$. Both vertices belong to the mono-connected component $E$, thus \textbf{\underline{Case IIIb}} is applied here. After insertion of edge $c$, the previously mono-connected component $E=(\{13,14,15,16\},9)$ now contains a cycle involving vertices $13$, $14$ and $15$. (i) We identify this cycle by considering the previous paths from vertices $14$ and $15$ to their articulation vertex $9$. These paths are $(14,13,9)$ and $(15,13,9)$, respectively. The first common vertex on this path is $13$; (ii) We create a new bi-connected component $G=(\{14,15\},13)$, containing all vertices of this cycle using the first common vertex $13$ as articulation vertex. We further remove these vertices except the articulation vertex $13$ from the mono-connected component $E$; the probability function $G.P(v)$ is initialized by sampling the reachability probabilities within $G$; and component $G$ is added to the list of children of $E$; (iii) Finally, orphans need to be collected. These are vertices in $E$, which have now become bi-connected to $Q$, because their (previously unique) path to their former articulation vertex $9$ crosses a new cycle. We find that one vertex, vertex $16$, had $15$ as the first removed vertex on its path to $9$. Thus, vertex $16$ is moved from component $E$ into a new mono-connected component $H=(\{16\},15)$, terminating this case. Summarizing, vertex $16$ in component $H$ now reports its information flow to vertex $15$ in component $G$, for which the information flow to articulation vertex $13$ in component $G$ is approximated using Monte-Carlo sampling. This information is then propagated analytically to vertex $9$ in component $E$, subsequently, the remaining flow that has been propagated all this way, is approximatively propagated to articulation vertex $6$ in component $C$, which allows to analytically compute the flow to articulation vertex $Q$.

For the last case, \textbf{\underline{Case IV}}, considering Figure \ref{fig:ct_example_graph_d}, where a new edge $d=(11,15)$ connected two vertices belonging to two different components $D$ and $E$. We start by identifying the cycle that has been created within the \emph{FT}, involving components $D$ and $E$, and meeting at the first common ancestor component $C$. For each of these components in the cycle $(D,C,E)$, one of the sub-cases of Case IV is used. For component $C$, we have that $C=C_{anc}$ is the common ancestor component, thus triggering \textbf{\underline{Case IVa}}. We find that both components $D$ and $E$ used vertex $9$ as their articulation vertex $v_{anc}$. Thus, the only cycle incurred in component $C$ is the (trivial) cycle $(9)$ from vertex $9$ to itself, which does not require any action. We initialize the new bi-connected component $\bigcirc=(\emptyset,\perp,9)$, which initially holds no vertices, and has no probability mass function computed yet (the operator $\perp$ can be read as null or not-defined) and uses $v_{anc}=9$ as articulation vertex. For component $D$, we apply \textbf{\underline{Case IVb}}, as $D$ is a bi-connected component, it becomes absorbed by a new bi-connected component $\bigcirc$, now having $\bigcirc=(\{10,11\},\perp,9)$. For the mono-connected component $E$ \textbf{\underline{Case IVc}} is used. We identify the path within $E$ that is now involved in a cycle, by using the path $(15,13,9)$ between the involved vertex $15$ to articulation vertex $9$. All nodes on this path are added to $\bigcirc$, now having $\bigcirc=(\{10,11,15,13\},\perp,9)$. Using the $splitTree(\cdot)$ operation similar to Case III, we collect orphans into new mono-connected components, creating $G=(\{14\},13)$ and $H=(\{16\},15)$ as children of $\bigcirc$. Finally, Monte-Carlo sampling is used to approximate the probability mass function $\bigcirc.P(v)$ for each $v\in\bigcirc.V$.

\section{Optimal Edge Selection}
\label{sec:algorithms}
The previous section presented the \emph{F-tree}, a data structure to compute the expected information flow in a probabilistic graph. Based on this structure, heuristics to find a near-optimal set of $k$ edges maximizing the information flow $\mbox{MaxFlow}(\mathcal{G},Q,k)$ to a vertex $Q$ (see Definition \ref{def:maxinfflow}) are presented in this section. Therefore, we first present a \emph{Greedy}-heuristic to iteratively add the locally most promising edges to the current result. Based on this \emph{Greedy} approach, we present improvements, aiming at minimizing the processing cost while maximizing the expected information flow.
\subsection{Greedy Algorithm}
\label{subsec:greedy}
Aiming to select edges incrementally, the \emph{Greedy} algorithm initially uses the probabilistic graph $G_0=(V,E_0=\emptyset,P)$, which contains no edges. In each iteration $i$, a set of candidate edges \emph{candList} is maintained, which contains all edges that are connected to $Q$ in the current graph $G_i$, but which are not yet selected in $E_i$. Then, each iteration selects an edge $e$ in addition maximizing the information flow to $Q$, such that $G_{i+1}=(V,E_i\cap e,P)$, where
\begin{equation}\label{eq:addedge}
e=\argmax_{e\in candList} E(\mbox{flow}(Q,(V,E_i\cap e,P))).
\end{equation}
For this purpose, each edge $e\in candList$ is probed, by inserting it into the current \emph{F-tree} using the insertion method presented in Section \ref{sec:ct}. Then, the gain in information flow incurred by this insertion is estimated by equation \ref{lem:sampling}. After $k$ iterations, the graph $G_k=(V,E_k,P)$ is returned.

\subsection{Component Memoization}
\label{subsec:memo}
We introduce an optimization reducing the number of computations for bi-connected components for which their reachability probabilities have to be estimated using Monte-Carlo sampling, by exploiting stochastic independence between different components in the \emph{F-tree}. During each \emph{Greedy}-iteration, a whole set of edges $candList$ is probed for insertion. Some of these insertions may yield new cycles in the \emph{F-tree}, resulting from cases III and IV. Using component \emph{Memoization}, the algorithm memoizes, for each edge $e$ in $candList$, the probability mass function of any bi-connected component $BC$ that had to be sampled during the last probing of $e$. Should $e$ again be inserted in a later iteration, the algorithm checks if the component has changed, in terms of vertices within that component or in terms of other edges that have been inserted into that component. If the component has remained unchanged, the sampling step is skipped, using the memoized estimated probability mass function instead. 
\subsection{Sampling Confidence Intervals}
\label{subsec:CI}
A Monte Carlo sampling is controlled by a parameter \emph{samplesize} which corresponds to the number of samples taken to approximate the information flow of a bi-connected component to its articulation vertex. In each iteration, we can reduce the amount of samples by introducing confidence intervals for the information flow for each edge $e\in candList$ that is probed. The idea is to prune the sampling of any probed edge $e$ for which we can conclude that, at a sufficiently large level of significance $\alpha$, there must exist another edge $e^\prime\neq e$ in $candList$, such that $e^\prime$ is guaranteed to have a higher information flow than $e$, based on the current number of samples only. To generate these confidence intervals, we recall that, following Equation \ref{eq:infflowsampling} the expected information flow to $Q$ is the sample-average of the sum of information flow of each individual vertex. For each vertex $v$, the random event of being connected to $Q$ in a random possible world follows a binomial distribution, with an unknown success probability $p$. To estimate $p$, given a number $S$ of samples and a number $0\leq s\leq S$ of 'successful' samples in which $Q$ is reachable from $v$, we borrow techniques from statistics to obtain a two sided $1-\alpha$ confidence interval of the true probability $p$. A simple way of obtaining such confidence interval is by applying the \emph{Central Limit Theorem} of Statistics to approximate a binomial distribution by a normal distribution.
\vspace{.1cm}
\begin{definition}[$\alpha$-Significant Confidence Interval]
Let $S$ be a set of possible graphs drawn from the probabilistic graph $\mathcal{G}$, and let $\hat{p}:=\frac{s}{S}$ be the fraction of possible graphs in $S$ in which $Q$ is reachable from $v$. With a likelihood of $1-\alpha$, the true probability $E({\updownarrow}(Q,v,\mathcal{G}))$ that $Q$ is reachable from $v$ in the probabilistic graph $G$ is in the interval
\begin{equation}\label{eq:clt}
\hat{p}\pm z\cdot\sqrt{\hat{p}(1-\hat{p})},
\end{equation}
where $z$ is the $100\cdot (1-0.5\cdot \alpha)$ percentile of the standard normal distribution. We denote the lower bound as $E_{lb}({\updownarrow}(Q,v,\mathcal{G}))$ and the upper bound as $E_{ub}({\updownarrow}(Q,v,\mathcal{G}))$. We use $\alpha=0.01$.
\end{definition}
\vspace{0.1cm}
To obtain a lower bound of the expected information flow to $Q$ in a graph $G$, we use the sum of lower bound flows of each vertex using Equation \ref{eq:infflowsampling} to obtain 
\vspace{0.1cm}
$$ E_{lb}(\mbox{flow}(Q,\mathcal{G}))=\sum_{v\in
V}E_{lb}({\updownarrow}(Q,v,\mathcal{G}))\cdot W(v) $$
as well as the upper bound
\vspace{0.1cm}
$$
E_{ub}(\mbox{flow}(Q,\mathcal{G}))=\sum_{v\in V}E_{ub}({\updownarrow}(Q,v,\mathcal{G}))\cdot W(v)
$$
Now, at any iteration $i$ of the \emph{Greedy} algorithm, for any candidate edge $e\prime\in candList$ having an information flow lower bounded by $lb:=E_{lb}(\mbox{flow}(Q,\mathcal{G}_i)\cup e)$, we prune any other candidate edge $e^\prime\in candList$ having an upper bound $ub:=E_{ub}(\mbox{flow}(Q,\mathcal{G}_i\cup e^\prime))$ iff $lb>ub$. The rational of this pruning is that, with a confidence of $1-\alpha$, we can guarantee that inserting $e\prime$ yields less information gain than inserting $e$. To ensure that the \emph{Central Limit Theorem} is applicable, we only apply this pruning step if at least 30 sample worlds have been drawn for both probabilistic graphs.
\subsection{Delayed Sampling}\label{subsec:DS}
For the last heuristic, we reduce the number of Monte-Carlo samplings that need to be performed in each iteration of the \emph{Greedy} Algorithm in Section \ref{subsec:greedy}. In a nutshell, the idea is that an edge, which yields a much lower information gain than the chosen edge, is unlikely to become the edge having the highest information gain in the next iteration. For this purpose, we introduce a delayed sampling heuristic. In any iteration $i$ of the \emph{Greedy} Algorithm, let $e$ denote the best selected edge, as defined in Equation \ref{eq:addedge}. For any other edge $e^\prime\in candList$, we define its potential $pot(e^\prime):=\frac{E(\mbox{flow}(Q,(V,E_i\cap e^\prime,P))}{E(\mbox{flow}(Q,(V,E_i\cap e,P))}$, as the fraction of information gained by adding edge $e\prime$ compared to the best edge $e$ which has been selected in an iteration. Furthermore, we define the cost $cost(e^\prime)$ as the number of edges that need to be sampled to estimate the information gain incurred by adding edge $e^\prime$. If the insertion of $e^\prime$ does not incur any new cycles, then $cost(e^\prime)$ is zero. Now, after iteration $i$ where edge $e^\prime$ has been probed but not selected, we define a sampling delay $$ d(e^\prime)=\lfloor log_c\frac{cost(e^\prime)}{pot(e^\prime)}\rfloor, $$ which implies that $e^\prime$ will not be considered as a candidate in the next $d$ iterations of the \emph{Greedy} algorithm of Section \ref{subsec:greedy}. This definition of delay, makes the (false) assumption that the information gain of an edge can only increase by a factor of $c>1$ in each iteration, where the parameter $c$ is used to control the penalty of having high sampling cost and having low information gain. As an example, assume an edge $e^\prime$ having an information gain of only $1\%$ of the selected best edge $e$, and requiring to sample a new bi-connected component involving $10$ edges upon probing. Also, we assume that the information gain per iteration (and thus by insertion of other edges in the graph), may only increase by a factor of at most $c=2$. We get $d(e^\prime)=\lfloor log_2\frac{10}{0.01}\rfloor=\lfloor log_2 1000\rfloor=9$. Thus, using delayed sampling and having $c=2$, edge $e^\prime$ would not be considered in the next nine iterations of the edge selection algorithm. It must be noted that this delayed sampling strategy is a heuristic only, and that no correct upper-bound $c$ for the change in information gain can be given. Consequently, the delayed sampling heuristic may cause the edge having the highest information gain not to be selected, as it might still be suspended. Our experiments show that even for low values of $c$ (i.e., close to $1$), where edges are suspended for a large number of iterations, the loss in information gain is fairly low.
\vspace{-0.35cm}
\section{Evaluation}
\label{sec:evaluation}
In this section, we empirically evaluate efficiency and effectiveness of our proposed solutions to compute a near-optimal subgraph of an uncertain graph to maximize the information flow to a source node $Q$, given a constrained number of edges, according to Definition \ref{def:maxinfflow}. As motivated in the introducing Section \ref{sec:intro}, two main application fields of information propagation on uncertain graphs are: \textbf{i)} information/data propagation in spatial networks, such as wireless networks or a road networks, and \textbf{ii)} information/belief propagation in social networks. These two types of uncertain graphs have extremely different characteristics, which require separate evaluation. A spatial network follows a locality assumption, constraining the set of pairwise reachable nodes to a spatial distance. Thus, the average shortest path between a pair of two randomly selected nodes can be very large, depending on the spatial distance. In contrast, a social network has no locality assumption, thus allowing to move through the network with very few hops. As a result, without any locality assumption, the set of nodes reachable in $k$-hops from a query node may grow exponentially large in the number of hops. In networks following a locality assumption, this number grows polynomial, usually quadratic (in sensor and road networks on the plane) in the range $k$, as the area covered by a circle is quadratic to its radius. Our experiments have shown, that the locality assumption - which clearly exists in some applications - has tremendous impact on the performance of our algorithms, including the baseline. Consequently, we evaluate both cases separately.

All experiments were evaluated on a system with Linux 3.16.7, x86\_64, Intel(R) Xeon(R) CPU E5,2609, 2.4GHz. All algorithms were implemented in Python $2.7$. Dependencies: NetworkX 1.11, numpy 1.13.1.
\subsection{Dataset Descriptions}
\label{subsec:datasets}
This section describes our employed uncertain graph datasets. For both cases, i.e. with locality assumption and no-locality assumption, we use synthetic and real datasets.

{\bfseries Synthetic Datasets: No locality assumption.} Our first model \emph{Erd\"os} is based on the idea of the \textit{Erd\"os-R\'{e}nyi model} \cite{erdos}, distributing edges independently and uniformly between nodes. Probabilities of edges are chosen uniformly in $[0,1]$ and weights of nodes are integers selected uniformly from $[0,10]$. It is known that this model is not able to capture real human social networks \cite{leskovec2005realistic}, due to the lack of modeling long tail distributions produced by ``social animals''. Thus, we use this data generation only in our first set of experiments and we employ real social network data later. 

{\bfseries Synthetic Datasets: Locality assumption.} We use two synthetic data generating scheme to generate spatial networks. For the first data generating scheme - denoted by \emph{partitioned} -, each vertex has the same degree $d$. The dataset is partitioned into $n=2 \cdot \frac{|V|}{d}$ partitions $P_0,...,P_{n-1}$ of size $d$. Each vertex in partition $P_i$ is connected to all and only vertices in the previous and next partition $P_{(i-1) \mod n}$ and $P_{(i+1) \mod n}$. This data generation allows to control the diameter of a resulting network, which is guaranteed to be equal to $n-1$.

For a more realistic synthetic data set - denoted as \emph{WSN} -, we simulate a wireless sensor network. Here, vertices have two spatial coordinates selected uniformly in $[0,1]$. Using a global parameter $\epsilon$, any vertex $v$ is connected to all vertices located in the $\epsilon$-distance of $v$ using Euclidean distance.
For both settings, the probabilities of edges are chosen uniformly in [0, 1]. 

{\bfseries Real Datasets: No locality assumption.} We use the \emph{social circles of Facebook dataset} published in \cite{leskovec2012learning}. This dataset is a snapshot of the social network of Facebook - containing a subgroup of $535$ users which form a social circle, i.e. a highly connected subgraph, having $10k$ edges. These users have excessive number of 'friends'. Yet, it has been discussed in  \cite{Wellman01doesthe} that the number of real friends that influence, affect and interact with an individual is limited. According to this result, and due to the lack of better knowledge which people of this social network are real friends, we apply higher edge probabilities uniformly selected in $[0.5;1.0]$ to $10$ random adjacent nodes of each user. Due to symmetry, an average user has $20$ such high probabilities 'close friends'. All other edges are assigned edge probabilities uniformly selected in $]0;0.5]$.

For our experiments on \emph{collaboration network data}, we used the computer science bibliography \emph{DBLP}. The structure of this dataset is such that if an author $v_i$ co-authored a paper with author $v_j$, where $i \neq j$, the graph contains a undirected edge $e$ from $v_i$ to $v_j$. If a paper is co-authored by $k$ authors this generates a completely connected (sub)graph (clique) on $k$ nodes. This dataset has been published in \cite{YangL12}. Probabilities on edges are uniformly distributed in $[0;1]$. The graph consists of $|V|=317,080$ vertices and $|E|=1,049,866$ edges.

Finally, we evaluated our methods also on the \emph{YouTube social network}, first published in \cite{mislove-2007-socialnetworks}. In this network, edges represent friendship relationships between users. The graph consists of $|V|= 1,134,890$ vertices and $|E|=2,987,624$ edges. Again, the probabilities on edges are uniformly distributed in $[0;1]$.

{\bfseries Real Datasets: Locality assumption.} For our experiments on \emph{spatial networks}, we used the road network of San Joaquin County\footnote{https://www.cs.utah.edu/~lifeifei/SpatialDataset.htm}, having $|V|=18,263$ vertices and $|E|=23,874$ edges. The vertices of the graph are road intersections and edges correspond to connections between them. In order to simulate real sensor nodes located at road intersections, we have connected vertices that are spatially distant from each other have a lower chance to successfully communicate. To give an example,  for two vertices having a distance of $a$ in meters, we set the communication probability to $e^{-0.001a}$. Thus, a $10m$, $100m$ and $1km$ distance will yield a probability of $e^{-0.01}=99\%$, $e^{-0.1}=90\%$, and $e^{-1}=36\%$, respectively.
\vspace{-0.4cm}
\begin{figure}[t]
  	\subfigure[Changing Graph Size with locality assumption]{
      \includegraphics[width=.23\textwidth,
       height=3cm]{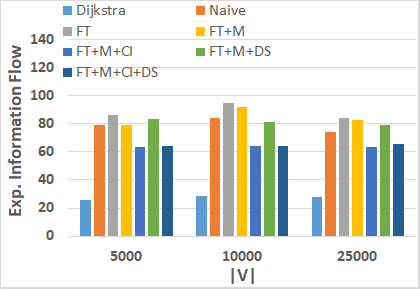}
      \includegraphics[width=.23\textwidth,
       height=3cm]{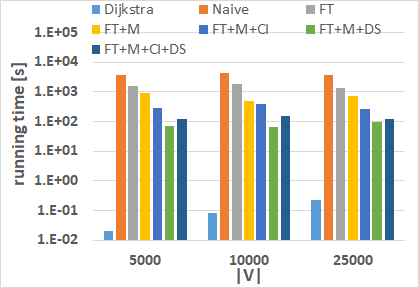}
    \label{fig:experiment_graphSize_with}
   }
   \subfigure[Changing Graph Size without locality assumption]{
       \includegraphics[width=.23\textwidth,
       height=3cm]{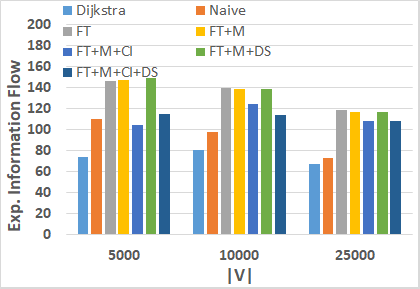}
       \includegraphics[width=.23\textwidth,
       height=3cm]{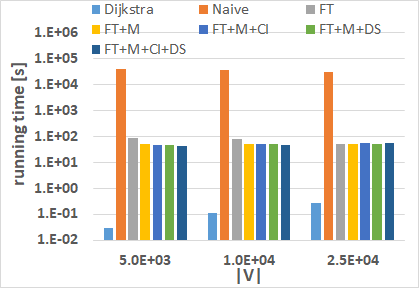}
    \label{fig:experiment_graphSize_without}
   }
  \caption{Experiments with changing graph size}\vspace{-0.4cm}
  \label{fig:graphSize}
\end{figure}

\subsection{Evaluated Algorithms}
The algorithms that we evaluate in this section are denoted and described as follows:

{\bfseries Naive} As proposed in \cite{Sampling06,emrich2012exploration} the first competitor \emph{Naive} does not utilize the strategy of component decomposition of Section \ref{sec:reachability} and utilizes a pure sampling approach to estimate reachability probabilities. To select edges, the \emph{Naive} approach chooses the locally best edge, as shown in Section \ref{sec:algorithms}, but does not use the \emph{F-tree} representation presented in Section \ref{sec:ct}. We use a constant Monte-Carlo sampling size of $1000$ samples.

{\bfseries Dijkstra} Shortest-path spanning trees \cite{sohrabi2000protocols} are used to interconnect a wireless sensor network (WSN) to a sink node. To obtain a maximum probability spanning tree, we proceed as follows: the cost $w(e)$ of each edge $e\in E$ is set to $w(e) = -\log(P(e))$. Running the traditional Dijkstra algorithm on the transformed graph starting at node $Q$ yields, in each iteration, a spanning tree which maximizes the connectivity probability between $Q$ and any node connected to $Q$ \cite{LinkDiscovery06}. Since, in each iteration, the resulting graph has a tree structure, this approach can fully exploit the concept of Section \ref{sec:reachability}, requiring no sampling step at all.

{\bfseries FT } employs the \emph{F-tree} proposed in Section \ref{sec:ct} to estimate reachability probabilities. To sample bi-connected components, we draw $1000$ samples for a fair comparison to \emph{Naive}. All following \emph{FT}-Algorithms build on top of \emph{FT}.

{\bfseries FT+M} additionally memoizes for each candidate edge $e$ the $pdf$ of the corresponding bi-connected component from the last iteration (cf. Section \ref{subsec:memo}).

{\bfseries FT+M+CI} further ensures that probing of an edge is stopped whenever another edge has a higher information flow with a certain degree of confidence, as explained in Section \ref{subsec:CI}.

{\bfseries FT+M+DS} instead tries to minimize the candidate edges in an iteration by leaving out edges that had a small information gain/cost-ratio in the last iteration (cf. Section \ref{subsec:DS}). Per default, we set the penalization parameter to $c=2$.

{\bfseries FT+M+CI+DS} is a combination of all the above concepts.
\begin{figure}[t]
  \subfigure[Changing Graph Density with locality assumption]{
      \centering
      \includegraphics[width=.24\textwidth,
       height=3cm]{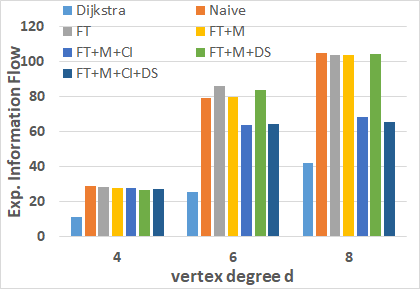}
      \includegraphics[width=.24\textwidth,
       height=3cm]{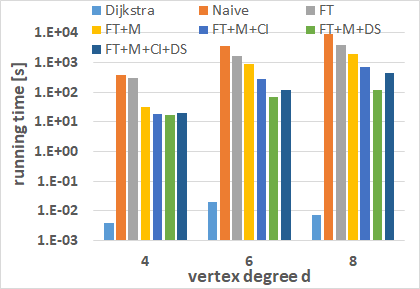}
    \label{fig:experiment_graphdensity_with}
   }
  \subfigure[Changing Graph Density without locality assumption]{

      \centering
       \includegraphics[width=.24\textwidth,
       height=3cm]{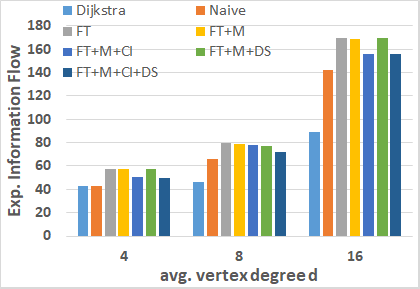}
       \includegraphics[width=.24\textwidth,
       height=3cm]{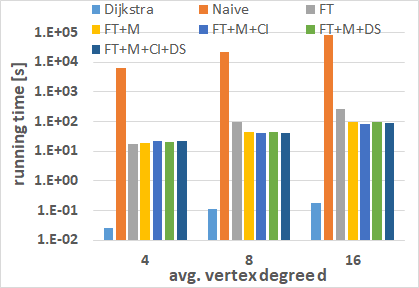}

    \label{fig:experiment_graphdensity_without}
    }
  \caption{Experiments with changing graph density}\vspace{-0.3cm}
  \label{fig:graphDensity}
\end{figure}

\subsection{Experiments on Synthetic Data}
\label{subsec:experiments_synthetic}
In this section, we employ randomly generated uncertain graphs. We generate graphs having no-locality-assumption using \emph{Erd\"os} graphs and having locality assumption using the \emph{partitioned} generation. Both generation approaches are described in Section \ref{subsec:datasets}. This data generation allows us to scale the topology of the uncertain graph $\mathcal{G}$ in terms of size and density. Unless specified otherwise, we use a graph size of $|V|=10,000$, a vertex degree of $6$ and a budget of edges $k=200$ in a all experiments on synthetic data.

{\bf Graph Size. }We first scale the size $|V|$ of the synthetic graphs. Figure \ref{fig:experiment_graphSize_with}
shows the information flow (left-hand-side) and running time (right-hand-side) for our synthetic data set following the locality assumption. First, we note that the \emph{Dijkstra}-based shortest-path spanning tree yields an extremely low information flow, far inferior to all other approaches. The reason is that a spanning tree allows no room for failure of edges: whenever any edge fails, the whole subtree become disconnected from $Q$. We further note that all other algorithms, including the \emph{Naive} one, are oblivious to the size of the network, in terms of information flow and running time. The reason is that, due to the locality assumption, only a local neighborhood of vertices and edges is relevant, regardless of the global size of the graph. Additionally, we see that the delayed sampling heuristic (\emph{DS}) yields a significant running time performance gain, whilst keeping the information flow constantly high. The combination of all heuristics (\emph{FT+M+CI+DS}) yields significant loss of information flow due to the pruning strategy of the confidence interval heuristic (\emph{CI}).
Figure \ref{fig:experiment_graphSize_without}, shows the performance in terms of information gain and running time for the \emph{Erd\"os} graphs having no locality assumption. We first observe that \emph{Dijkstra} and \emph{Naive} yield a significantly lower information flow than our proposed approaches. For \emph{Dijkstra}, this result is again contributed to the constraint of constructing a spanning tree, and thus not allowing any edges to connect the flow between branches. For the \emph{Naive} approach, the loss in information flow requires a closer look. This approach samples the whole graph only $1000$ times, to estimate the information flow. In contrast, our \emph{F-tree} approach samples each individual bi-connected component $1000$ times. Why is the later approach more accurate? A first, informal, explanation is that, for a constant sampling size, the information flow of a small component can be estimated more accurately than for a large component. Intuitively, sampling two independent components $n$ times each, yields a total of $n^2$ combinations of samples of their joint distribution. More formally, this effect is contributed to the fact that the variance of the sum of two random variables increases as their correlation increases, since $Var(\sum_{i=1}^n X_i)=\sum_{i=1}^n Var(X_i)+2\sum_{1\leq i<j\leq n} Cov(X_i,X_j)$ \cite{navidi2006statistics}. Furthermore, the \emph{Naive} approach also incurs an approximation error for mono-connected components, for which all \emph{F-tree} (\emph{FT}) approaches compute the exact flow analytically. We further see that the \emph{Naive} approach, which has to sample the whole graph, is by far the most inefficient. On the other end of the scope, the \emph{Dijkstra} approach, which is able to avoid sampling entirely by guaranteeing a single mono-connected component, is the fastest in all experiments, but at the cost of information flow. 
We also see that in Figure \ref{fig:experiment_graphSize_without} all algorithm stay in the same order of magnitude in their running time and information flow as the graph increases. This is due to the fact that in this experiment we stay constant in the average vertex degree, i.e. $deg(v_i)\approx 10$ for all $v_i \in V$. We also observe that the \emph{CI} heuristic yields an overhead of computing the intervals whilst losing information gain due to its rigorous pruning strategy.
\begin{figure}[ht]
  \subfigure[Changing budget k with locality assumption]{
      \centering
      \includegraphics[width=.24\textwidth,
      height=3cm]{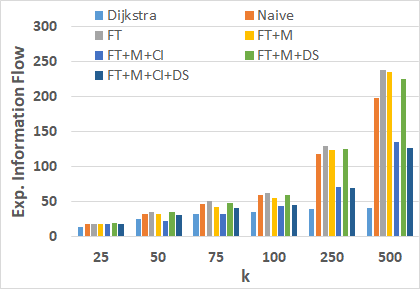}
      \includegraphics[width=.24\textwidth,
      height=3cm]{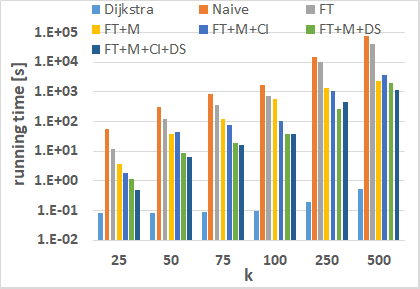}
    \label{fig:experiment_budget_with}
   }
  \subfigure[Changing budget k without locality assumption]{
      \centering
       \includegraphics[width=.24\textwidth,
      height=3cm]{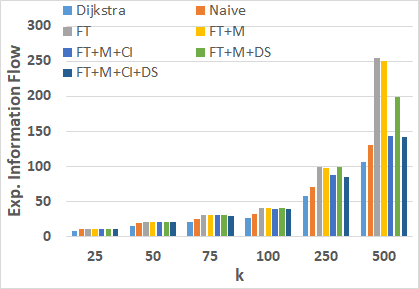}
       \includegraphics[width=.24\textwidth,
      height=3cm]{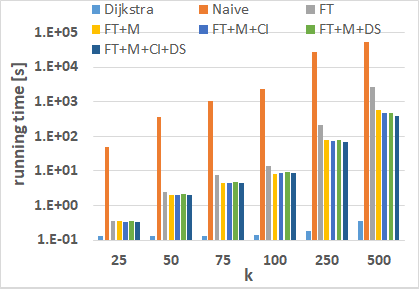}
    \label{fig:experiment_budget_without}
    }
  \caption{Experiments with changing Budget}\vspace{-0.3cm}
  \label{fig:budget}
\end{figure}

\begin{figure}[ht]
  \subfigure[WSN $\epsilon=0.05$]{
      \centering
        \includegraphics[width=.24\textwidth,
        height=3cm]{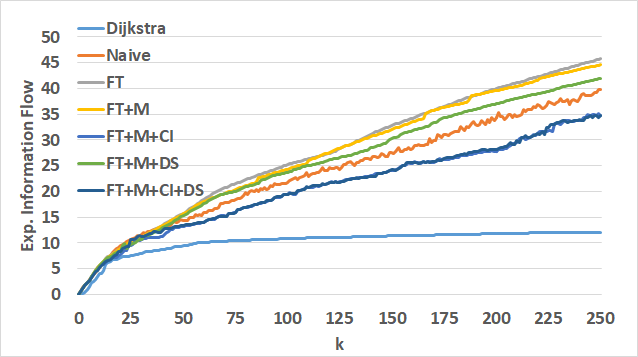}
        \includegraphics[width=.24\textwidth,
        height=3cm]{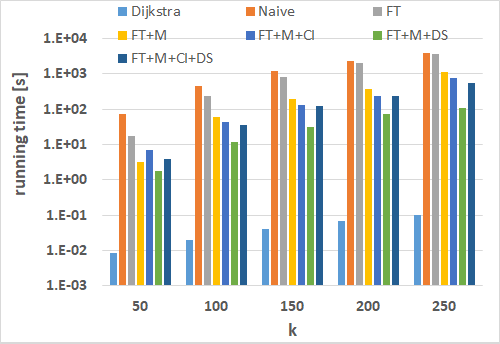}
    \label{fig:experiment_wsn_0-05}
    }
  \subfigure[WSN $\epsilon=0.07$]{
      \centering
      \includegraphics[width=.24\textwidth,
        height=3cm]{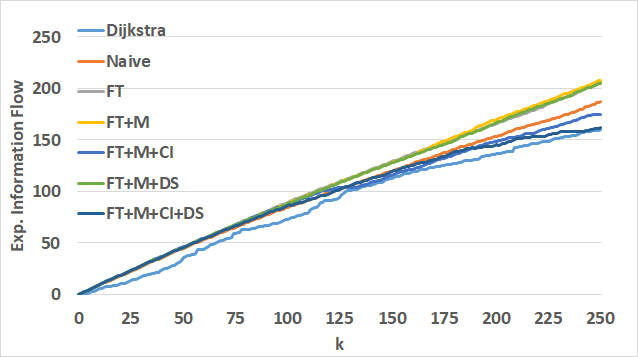}
      \includegraphics[width=.24\textwidth,
        height=3cm]{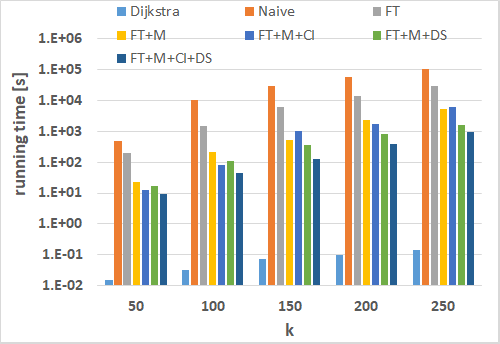}
    \label{fig:experiment_wsn_0-07}
    }\vspace{-0.25cm}
  \caption{Experiments in synthetic Wireless Sensor
  Networks}\vspace{-0.5cm}
  \label{fig:experiment_wsn}
\end{figure}
\begin{figure*}[t]
  \hfil
  \subfigure[San Joaquin County \scriptsize Road~Network]{
    \parbox{0.225\textwidth}{
      \centering
        \includegraphics[width=.24\textwidth,
        height=3cm]{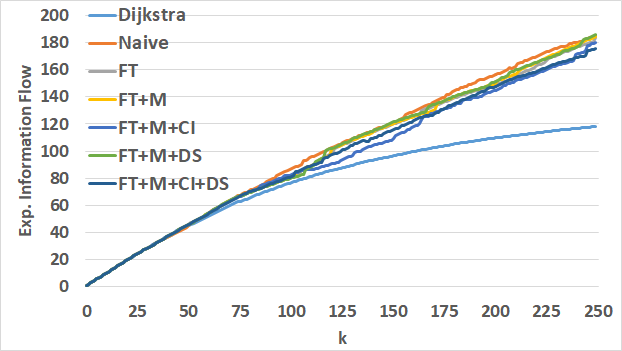}
        \includegraphics[width=.24\textwidth,
        height=3cm]{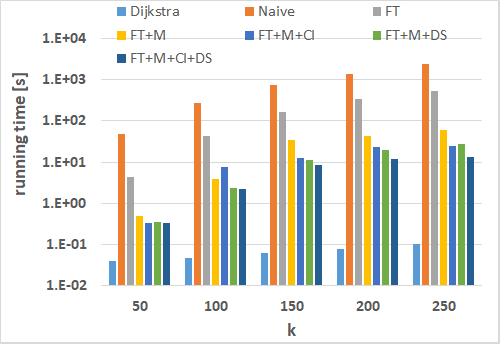}
        }
    \label{fig:experiment_sanJoaquin}
    }
  \hfil
  \subfigure[Circle of Friends - Facebook]{
    \parbox{0.225\textwidth}{
      \centering
        \includegraphics[width=.24\textwidth,
        height=3cm]{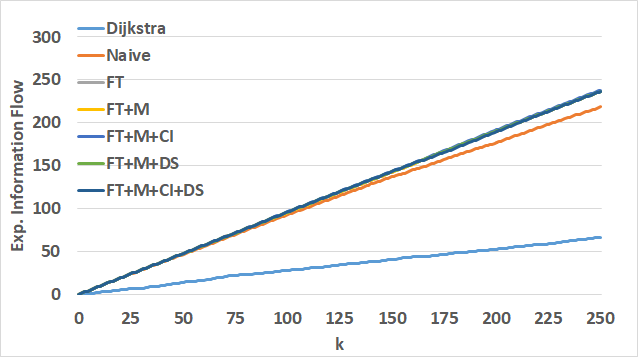}
        \includegraphics[width=.24\textwidth,
        height=3cm]{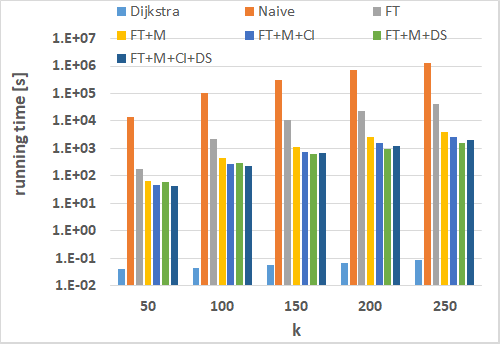}
        }
    \label{fig:experiment_facebook}
    }
  \hfil
  \subfigure[DBLP]{
    \parbox{0.225\textwidth}{
      \centering
        \includegraphics[width=.24\textwidth,
        height=3cm]{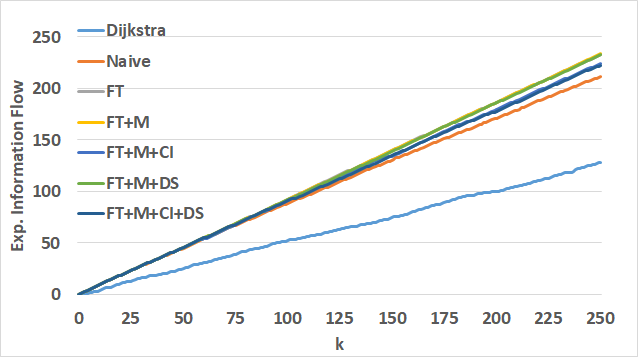}
        \includegraphics[width=.24\textwidth,
        height=3cm]{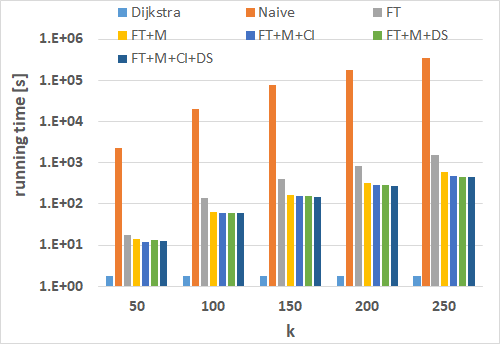}
        }
    \label{fig:experiment_dblp}
    }
  \hfil
  \subfigure[YouTube]{
    \parbox{0.225\textwidth}{
      \centering
        \includegraphics[width=.24\textwidth,
        height=3cm]{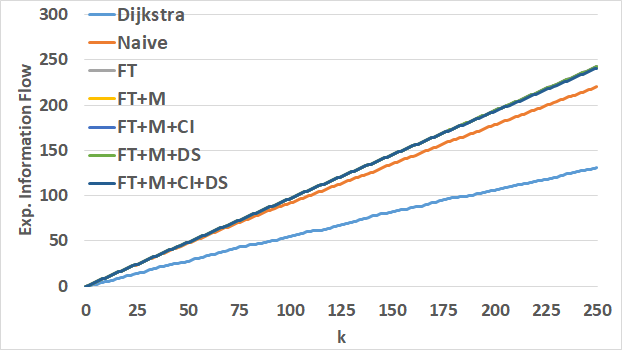}
        \includegraphics[width=.24\textwidth,
        height=3cm]{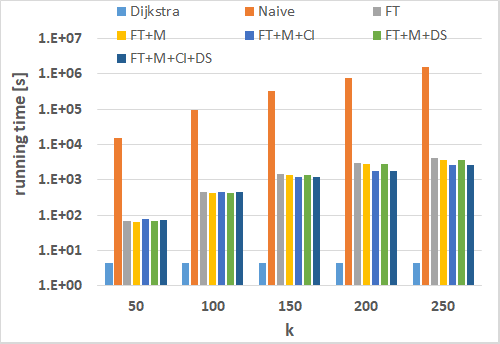}
        }
    \label{fig:experiment_yt}
    }\vspace{-0.18cm}
  \caption{Experiments on Real World Datasets}\vspace{-0.45cm}
  \label{fig:experiment_realWorld}
\end{figure*}

{\bf Graph Density.} In this experiment, we scale the average degree of vertices. In the case of graphs following the locality assumption, the gain in information flow of all proposed solutions compared to \emph{Dijkstra} is quite significant as shown in Figure \ref{fig:experiment_graphdensity_with}, particularly when the degree of vertices is low. 
This is the case in road networks, but also in most sensor and ad hoc networks. The reason is that, in such case, spanning trees gain quickly in height as edges are added, thus incurring low-probability paths that require circular components to connect branches to support the information flow. For larger vertex degrees, the optimal solutions become more star-like, thus becoming more tree-like.  For a small vertex degree, we observe also that the same bi-connected-components are occurring in consecutive iterations resulting in a running time gain for the memoization approach. As the complexity of the graph grows, the gap between the \emph{FT} and \emph{FT+M} shrinks as more candidates result in an increased number of possibilities where bi-connected components can occur and thus make cached results for bi-connected components obsolete.
The results shown in Figure \ref{fig:experiment_graphdensity_without} indicate that the \emph{Dijkstra} approach is able to find higher quality result in graphs without locality assumption for small (average) vertex degrees. This is contributed to the fact that in graphs with this setting, the optimal result will be almost tree-like, having only a few inter-branch edges. 
The algorithms \emph{FT+M+CI} and \emph{FT+M+CI+DS} yield a trade-off between running time and accuracy, i.e. we observe a slightly loss in information gain coming along with a better running time for a setting with a larger (average) vertex degree.


{\bf Scaling of parameter $k$.} In the next experiments, shown in Figure \ref{fig:budget}, we show how the budget $k$ of edges affects the performance of the proposed algorithms. In the case of a network following the locality assumption, we observe in Figure \ref{fig:experiment_budget_with} that the overall information gain per additional edge slowly decreases. This is clear, since in average, the hop distance to $Q$ increases as more edges are added, increasing the possible links of failure, thus decreasing the likelihood of propagating a nodes information to $Q$. We observe that the effectiveness of \emph{Dijkstra} in the locality setting quickly deteriorates, since the constraint of returning a tree structure leads to paths between $Q$ and other connected nodes that become increasingly unlikely, making the \emph{Dijkstra} approach unable to handle such settings. 
Here, the memoization heuristic $M$ perform extremely well. A sever loss of information gain is observed when running \emph{FT+M+CI} and \emph{FT+M+CI+DS} due to its pruning policy. Later one is the best in terms of running time. 



In contrast, using a network following no locality assumption in Figure \ref{fig:experiment_budget_without}, we see that both \emph{Dijkstra} and \emph{Naive} yield an low information gain for a large budget $k$. For \emph{Dijkstra}, the reason is that for large values of $k$, the depths of the spanning tree, which is lower bounded by $log_d k$, incurs longer paths without any backup route in the case of a single failure. For the \emph{Naive} approach, the low information gain is contributed to the high variance of sampling the information flow of the whole graph for each edge selection. Further, we see that the \emph{Naive} approach further suffers from an extreme running time, requiring to re-sample the whole graph in each iteration. The $F$-tree in combination with the memoization give a consistently high information gain while having a low running time. The heuristics suffering from a loss in information gain yield a slightly better running time.



{\bf Synthetic Wireless Sensor Networks (WSN).} In this experiments, we simulate real world wireless sensor networks (WSN). We embed a number of vertices - here $|V|=1,000$ - according to a uniform distribution in a spatial space $[0;1] \times [0;1]$. For each vertex, we observe adjacent vertices being in its proximity which is regulated by an additional parameter $\epsilon$. Figure \ref{fig:experiment_wsn} shows the results. We observe nearly the same behavior as in Figure \ref{fig:experiment_graphdensity_with}. As the parameter $\epsilon$ is a regulator for the graph's interconnectivity, we observe again a fair trade-off of information gain and running time for the proposed heuristics.
By increasing the parameter $\epsilon$, hence, simulating dense graphs, the gap between \emph{Dijkstra} and the \emph{F-tree} approaches is reduced. For these datasets, we can also observe the benefit of \emph{FT+M+CI+DS} which still identifies a high information gain whilst reducing the running time, as the number of candidates are reduced, respectively, we can prune candidates in earlier stages of each iteration.

{\bfseries Parameter $c$.} We also evaluated the penalization parameter $c$ of the delayed sampling heuristics and summarize the results. In all our evaluated settings, ranging from $1.01\geq c \geq 16$, the running time consistently decreases as $c$ is decreased, yielding a factor of $2$ to $10$ speed-up for $c=1.2$, depending on the dataset, and a multi-orders of magnitude speed-up for $c=1.01$. Yet, for $c<1.2$ we start to observe a significant loss of information flow. For the extreme case of $c=1.01$, the information flow became worse than \emph{Dijkstra}, as edges become suspended unreasonable long, choosing edges nearly arbitrarily. For the default setting of $c=2$ used in all previous evaluations, the \emph{delayed sampling heuristics} showed insignificant loss of information, but yielding a better running time.
\vspace{-0.25cm}
\subsection{Experiments on Real World Data}
Our first real world data experiments uses the \emph{San Joaquin County Road Network}. As road networks are of very sparse nature, and follow a strong locality assumption, our approaches outperforms \emph{Dijkstra} significantly as $k$ is scaled to $k=250$. Thus, \emph{Dijkstra} is highly undesirable as budget is wasted without proper return in information flow. In this setting, following the locality assumption, we see that all heuristics yield a significant run time performance gain, while the information flow remains similar for all heuristics.


In the next experiment, we employ the \emph{social circles of friends} dataset, an extremely dense network with no locality assumption, where most pairs of nodes are connected. As described in Section \ref{subsec:datasets}, each vertex in this graph only has ten high-probability links, whereas all other nodes have a lower probability.
Figure \ref{fig:experiment_facebook} shows that 
\emph{Dijkstra} yields a most significant loss of information, as it is forced to quickly build a large-height tree to maintain high probability edges. Further, we see that the memoization heuristic yields a significant running time improvement of about one order of magnitude. We note that in such dense setting, heuristics \emph{CI} and \emph{DS} show almost no effect in both, runtime and information flow. 

Figure \ref{fig:experiment_dblp} shows similar results on the \emph{DBLP} collaboration network dataset, a sparse network which follows no locality assumption. Again, we observe a loss of potential information flow for \emph{Dijkstra} as $k$ increases. 

Finally, we observe similar behavior of all approaches on a bigger graph such as the \emph{YouTube social network}, which refers to a sparse setting with no locality assumption. Figure \ref{fig:experiment_yt} shows the results. As in the other settings, we can observe an extremely low information flow of \emph{Dijkstra}, and an extreme running time of the \emph{Naive} approach. It is interesting that in this setting, the memoization approach \emph{FT+M} yields only a minimal gain in running time, like the other heuristics. Fortunately, none of these heuristics shows a significant loss of information flow.
\vspace{-0.20cm}
\subsection{Experimental Evaluation: Summary}\vspace{-0.10cm}
To summarize our experimental results, we reiterate the shortcomings of the naive solutions, and briefly discuss which of the heuristics are best used in which setting.

{\bfseries Naive:} Our \emph{Naive} competitor, which applies a Greedy edge selection (c.f. Sec. \ref{sec:algorithms}) but does not use the \emph{F-tree}, is multiple orders of magnitude slower than our other approaches in all real-data experiments (c.f. Fig. \ref{fig:experiment_realWorld}). Further, large sampling errors also yield a significantly lower information flow in most settings.

{\bfseries Dijkstra:} A \emph{Dijkstra}-based spanning tree algorithm runs extremely fast, but at the cost of an extreme loss of information, yielding low information flow. The information loss is particularly high for social networks (e.g. Fig. \ref{fig:experiment_facebook}), where cycles are required to increase the odds of connecting a distant node to the source.

{\bfseries FT:} Employing the \emph{F-tree} proposed in Section \ref{sec:ct} maximizes the information flow. Compared to the \emph{Naive} approach, smaller partitions need to be sampled yielding smaller sampling variation while being multiple orders of magnitude faster. 

{\bfseries FT+M:} The memoization heuristic technique described in Section \ref{subsec:memo} was shown to be simple and effective. It yields vast reduction in running time of up to one order of magnitude on real-data (see Fig. \ref{fig:experiment_realWorld}), while showing no notable detriment to the information flow.

{\bfseries FT+M+CI:} Employing confidence intervals as described in Section \ref{subsec:CI} has shown a significant improvement in running time on spatial networks following the locality assumption (c.f. Figures \ref{fig:experiment_graphSize_with}, \ref{fig:experiment_graphdensity_with}, and \ref{fig:experiment_sanJoaquin}). However, this heuristic yields no improvement (and often has a detrimental effect) in settings without locality assumptions such as in social networks (c.f. Fig. \ref{fig:experiment_facebook}-\ref{fig:experiment_yt}). This heuristic should not be employed in such settings.

{\bfseries FT+M+DS:} The delayed sampling heuristic presented in Section \ref{subsec:DS} yields an improvement in running time in networks following the locality assumption. This improvement is especially large in cases having a high vertex degree (c.f. Figure \ref{fig:experiment_graphdensity_with}). However, in social networks which do not follow the locality assumption, the gain of this heuristic is often marginal (c.f. Fig. \ref{fig:experiment_facebook}-\ref{fig:experiment_yt}). Yet, this heuristic comes at minimal loss of information flow, such that it is not detrimental to enable it by default.

{\bfseries FT+M+CI+DS:} The combination of all heuristics inherits the problems of FT+M+CI and FT+M+DS for the cases without locality assumption. But for the cases with locality assumption, our experiments on real world data show that in most cases the combination of all heuristic achieves significant lower running time compared to setting where we apply each of the heuristics separately proofing the importance of each proposed heuristic.   
\vspace{-0.45cm}
\section{Conclusions}\vspace{-0.2cm}
In this paper we discussed solutions for the problem of maximizing the information flow in an uncertain graph given a fixed budget of $k$ communication edges. We identified two \emph{NP}-hard subproblems that needed heuristical solutions: \textbf{(i)} Computing the expected information flow of a given subgraph; and \textbf{(ii)} selecting the optimal $k$-set of edges. For problem \emph{(i)}, we developed an advanced sampling strategy that only performs an expensive (and approximative) sampling step for parts of the graph for which we can not obtain an efficient (and exact) analytic solution. For problem \emph{(ii)}, we propose our \emph{F-tree} representation of a graph $G$, which keeps track of \emph{bi-connected components} - for which sampling is required to estimate the information flow - and \emph{mono-connected components} - for which the information flow can be computed analytically. On the basis of the \emph{F-tree} representation, we introduced further approaches and heuristics to handle the trade-off between effectiveness and efficiency. Our evaluation shows that these enhanced algorithms are able to find high quality solutions (i.e., $k$-sets of edges having a high information flow) in efficient time, especially in graphs following a locality assumption, such as road networks and wireless sensor networks.\vspace{-0.2cm}

\section{Acknowledgement}
This paper is supported by the Deutsche Forschungsgemeinschaft (DFG) under grant number RE 266/5-1.

\ifCLASSOPTIONcaptionsoff
  \newpage
\fi



%

\bibliographystyle{abbrv}
\bibliography{abbrev,bibliography}


%


\vspace{-1.25cm}
\begin{IEEEbiography}[{\includegraphics[width=1in,height=1in,clip,keepaspectratio]{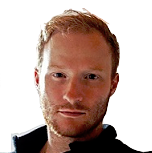}}]{Christian
Frey} \noindent is a research fellow at the Institute for Informatics at the Ludwig-Maximilians-Universit\"{a}t M\"{u}nchen,
Germany. His research interests include query processing in (uncertain) graph databases,
network analysis on large heterogeneous
information networks and machine learning approaches on (heterogeneous) information networks/Knowledge Graphs.
\end{IEEEbiography}
\vspace{-1.15cm}
\begin{IEEEbiography}[{\includegraphics[width=1.1in,height=1.1in,clip,keepaspectratio]{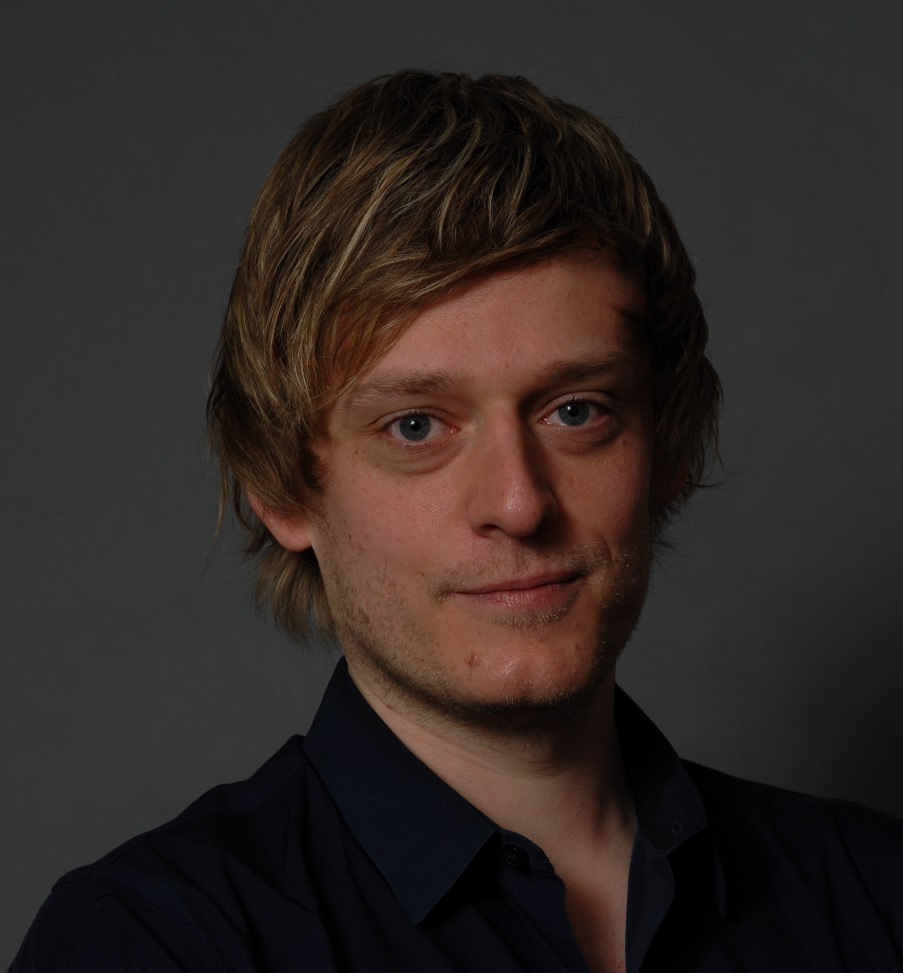}}]{Andreas
Z\"{u}fle} \noindent is an assistant professor with the department of
Geography and Geoinformation Science at George Mason University. Dr. Z\"{u}fle’s research expertise includes big spatial data, spatial data mining, social network mining, and uncertain database management.Since 2016, Dr. Z\"{u}fle research has received more than \$2,000,000 in research grants by the National Science Foundation (NSF) and the Defense Advanced Research Projects Agency (DARPA). Since 2011, Dr. Z\"{u}fle has published more than 60 papers in refereed conferences and journals having an h-index of 16. 
\end{IEEEbiography}
\vspace{-1.15cm}
\begin{IEEEbiography}[{\includegraphics[width=1.1in,height=1.1in,clip,keepaspectratio]{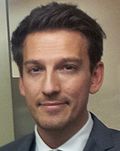}}]{Tobias
Emrich} \noindent received his PhD in Computer Science from LMU,
Munich in 2013. He then did his Post-Doc at the Integrated Media Systems Center
at the University of Southern California in 2014. In 2015 he went back to LMU to
become the Director of the Data Science Lab. Since then he started and led many
industry collaborations on Data Science topics with companies such as Siemens,
Volkswagen, Roche, and IAV. His research interest include similarity search and
data mining in spatial, temporal, uncertain and dynamic graph databases. To date
he has more than 40 publications in refereed conferences.
\end{IEEEbiography}
\vspace{-1.15cm}
\begin{IEEEbiography}[{\includegraphics[width=1.1in,height=1.1in,clip,keepaspectratio]{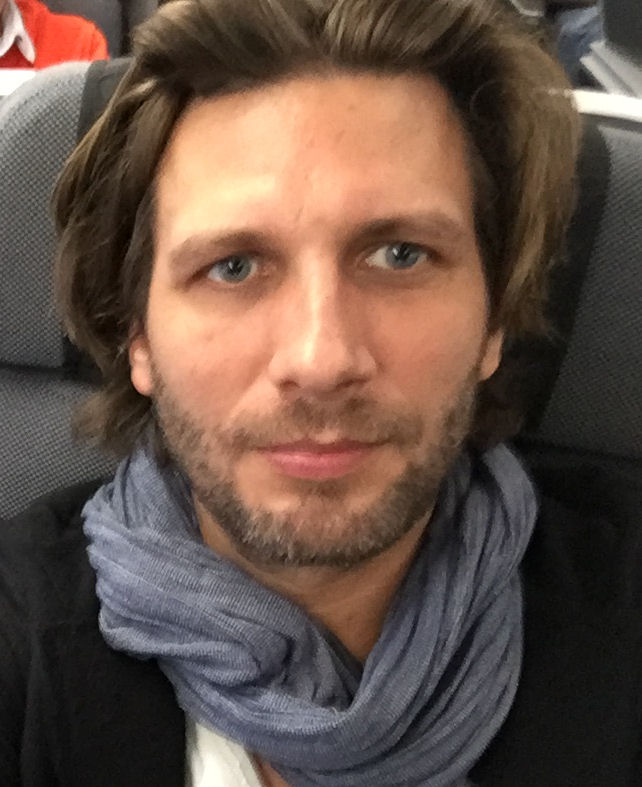}}]{Matthias
Renz} \noindent is an associate professor at the Computational and
Data Science Department at George Mason University. He received his PhD in
computer science at the Ludwig-Maximilians-Universit\"{a}t (LMU) Munich 2006,
and his habilitation 2011. His main research topics are data science, scientific and spatial databases, data mining and uncertain databases. To date, he has more than 120 peer-reviewed publications that in total received over 2200 citations.
\end{IEEEbiography}




\end{document}